\documentclass[10pt,conference]{ieeeconf}
\usepackage{textcomp}
\usepackage{url}            % simple URL typesetting
\usepackage{amsfonts}       % blackboard math symbols
\usepackage{xcolor}

\usepackage{amsmath, amssymb, mathtools, amsthm}
\usepackage{graphicx}
\usepackage{bm}
\usepackage{float}
\usepackage{color}
\usepackage{epstopdf}
\usepackage[ruled,vlined,linesnumbered,noend]{algorithm2e}

\makeatletter
\let\NAT@parse\undefined
\makeatother

\usepackage{hyperref}
\makeatletter
\hypersetup{colorlinks=true}
\AtBeginDocument{\@ifpackageloaded{hyperref}
  {
  \def\@linkcolor{blue}
  \def\@anchorcolor{red}
  \def\@citecolor{red}
  \def\@filecolor{red}
  \def\@urlcolor{black}
  \def\@menucolor{red}
  \def\@pagecolor{red}
\begingroup
  \@makeother\`%
  \@makeother\=%
  \edef\x{%
    \edef\noexpand\x{%
      \endgroup
      \noexpand\toks@{%
        \catcode 96=\noexpand\the\catcode`\noexpand\`\relax
        \catcode 61=\noexpand\the\catcode`\noexpand\=\relax
      }%
    }%
    \noexpand\x
  }%
\x
\@makeother\`
\@makeother\=
}{}}
\makeatother

\newtheorem{Theorem}{Theorem}

\newtheorem{Lemma}{Lemma}

\newtheorem{Problem}{Problem}

\newtheorem*{Problem*}{Problem}

\newtheorem{Remark}{Remark}

\newtheorem{Proposition}{Proposition}

\newtheorem{Assumption}{Assumption}

\newtheorem{Definition}{Definition}

\IEEEoverridecommandlockouts

\def\BibTeX{{\rm B\kern-.05em{\sc i\kern-.025em b}\kern-.08em
    T\kern-.1667em\lower.7ex\hbox{E}\kern-.125emX}}

\begin{document}

% \title{\LARGE{\bf Secure by Design: Finding Bounds to Prevent Damages from Attacks}}

\title{\LARGE{\bf Sampling based Computation of Viability Domain to \\ Prevent Safety Violations by Attackers}}

\author{Kunal~Garg \and
% \IEEEmembership{Student Member, IEEE}, 
Alvaro~A.~Cardenas \and Ricardo~G.~Sanfelice
% \IEEEmembership{Senior Member, IEEE}
\thanks{
The research was sponsored by the Army Research Office and was accomplished under Grant Number W911NF-20-1-0253. The views and conclusions contained in this document are those of the authors and should not be interpreted as representing the official policies, either expressed or implied, of the Army Research Office or the U.S. Government. The U.S. Government is authorized to reproduce and distribute reprints for Government purposes, notwithstanding any copyright notation herein.
% This paragraph of the first footnote will contain the date on which you submitted your brief for review. It will also contain support information, including sponsor and financial support acknowledgment. For  example, ``This work was supported in part by the U.S. Department of  Commerce under Grant BS123456.'' 
}
}
\maketitle
\begin{abstract}
This paper studies the security of cyber-physical systems under attacks. Our goal is to design system parameters, such as a set of initial conditions and input bounds so that it is secure by design. To this end, we propose new sufficient conditions to guarantee the safety of a system under adversarial actuator attacks. Using these conditions, we propose a computationally efficient sampling-based method to verify whether a set is a viability domain for a general class of nonlinear systems. In particular, we devise a method of checking a modified barrier function condition on a finite set of points to assess whether a set can be rendered forward invariant. Then, we propose an iterative algorithm to compute the set of initial conditions and input constraint set to limit what an adversary can do if it compromises the vulnerable inputs. Finally, we utilize a Quadratic Program approach for online control synthesis. 
\end{abstract}

% \begin{IEEEkeywords}
% Distributed algorithms, Optimization, Nonlinear control systems,  Power generation dispatch
% \end{IEEEkeywords}

\section{Introduction}
Security has become one of the most critical problems in the field of Cyber-Physical Systems (CPS), as illustrated by several incidents of attacks that happened in the past few years~\cite{lee2014german,oueslati2019comparative}. There are two types of security mechanisms for protecting CPS~\cite{cardenascyber} i) proactive, which considers design choices deployed in the CPS \emph{before} attacks, and ii) reactive, which take effect after an attack is detected.%, and they reconfigure the system online to minimize the impact of the attack. 
%For safe and reliable operations of CPS, it is essential that the designed control strategies are resilient against various possible attacks, such as attacks on the system's sensors, actuators, communication systems, or decision-making modules.   

%As discussed in \cite{cardenascyber}, there are two ways of mitigating attacks, i) proactive, which considers design choices deployed in the CPS before any attack, and ii) reactive, which take effect once an attack has been detected and they reconfigure the system online in order to minimize the impact of the attack. 
While reactive methods are less conservative than proactive mechanisms, they heavily rely on fast and accurate attack detection mechanisms. 
%On the other hand, proactive methods, such as designing the system, though conservatively, do not rely on such additional mechanisms. 
Although there is a plethora of work on attack detection for CPS~\cite{choi2018detecting,renganathan2020distributionally},  it is generally possible to design a stealthy attack such that the system behavior remains close to its expected behavior, thus evading attack-detection solutions~\cite{urbina2016limiting}. Intrusion detection systems also produce a large number of false positives, which can lead to a large operational overhead of security analysts dealing with irrelevant alerts~\cite{cardenas2006framework}. On the other hand, a proactive method can be more effective in practice, particularly against stealthy attacks. 
%might become a significant operational nuisance during normal operations. %Finally, as argued in \cite{wu2021secure}, reactive schemes have limited broader applications since an attacker can monitor and acquire the knowledge of the system's defense mechanisms and launch an effective attack.   
Attacks on a CPS can disrupt the natural operation of the system. 
% Among other desirable properties of a control system,
One of the most desirable system properties is \textit{safety}, i.e., the system does not go out of a safe zone. Safety is an essential requirement, violation of which can result in failure of the system, loss of money, or even loss of human life, particularly when a system is under attack \cite{al2018cyber}.

In most practical problems, safety can be realized as guaranteeing forward-invariance of a safe set. Control barrier function (CBF) based approach \cite{ames2017control} to guarantee forward invariance of the safe region has become very popular in the last few years since a safe control input can be efficiently computed using a Quadratic Program (QP) with CBF condition as the constraint. Most of the prior work on safety using CBFs, e.g., \cite{ames2017control}, assumes that the \textit{viability} domain, i.e., the set of initial conditions from which forward invariance of the safe set can be guaranteed, is known. In practice, it is not an easy task to compute the viability domain for a nonlinear control system. Optimization-based methods, such as Sum-of-Squares (SOS) techniques, have been used in the past to compute this domain (see \cite{wang2018permissive}). However, SOS-based approaches are only applicable to systems whose dynamics is given by polynomial functions, thus limiting their applications. Another method popularly used in the literature for computing the viability domain is Hamilton-Jacobi (HJ) based reachability analysis, see, e.g., \cite{choi2021robust}. However, such an analysis is computationally expensive, particularly for higher dimensional systems. We propose a novel sampling-based method to compute the viability domain for a general class of nonlinear control systems to overcome these limitations.

% In this paper, we propose a proactive method for designing a system resilient to actuator attacks.
In this work, we consider a general class of nonlinear systems under actuator attacks and propose a method of computing a set of initial conditions and an input constraint set such that the system remains \textit{secure by design}. In particular, we consider actuator manipulation, where an attacker can assign arbitrary values to the input signals for a subset of the actuators in a given bound. We consider the property of safety with respect to an unsafe set and propose sufficient conditions using sampling of the boundary of a set to verify whether the set is a \textit{viability domain} under attacks. Using these conditions, we propose a computationally tractable algorithm to compute the set of initial conditions and the input constraint set such that the system's safety can be guaranteed under attacks. In effect, our proposed method results in a secure-by-design system that is resilient against actuator attacks. Finally, we leverage these sets in a QP-based approach with provable feasibility for real-time online feedback synthesis. The contributions of the paper are summarized below:
\begin{itemize}
    % \item[1)] In contrast to most of the prior work on security, where forward reachability is used for safety, we utilize CBFs to derive sufficient conditions for safety under a class of attacks on system inputs;
    \item[1)] We present sampling-based sufficient conditions to assess whether a given set can be rendered forward invariant for a general class of nonlinear system. To the best of the authors' knowledge, this is the first work utilizing sampled-data approach for computing a viability domain;
    \item[2)] We present a novel iterative algorithm to compute a viability domain and an input constraint set to guarantee system safety under attacks. Unlike \cite{wang2018permissive,choi2021robust}, the proposed method is applicable for general nonlinear control systems and is scalable with the system dimension.
\end{itemize}

\noindent \textit{Notation}: Throughout the paper, $\mathbb R$ denotes the set of real numbers and $\mathbb R_+$ denotes the set of non-negative real numbers. We use $|x|$ to denote the Euclidean norm of a vector $x\in \mathbb R^n$. We use $\partial S$ to denote the boundary of a closed set $S\subset \mathbb R^n$ and $\textrm{int}(S)$ to denote its interior and $|x|_S = \inf_{y\in S}|x-y|$, to denote the distance of $x\in \mathbb R^n$ from the set $S$. The Lie derivative of a continuously differentiable function $h:\mathbb R^n\rightarrow\mathbb R$ along a vector field $f:\mathbb R^n\rightarrow\mathbb R^m$ at a point $x\in \mathbb R^n$ is denoted as $L_fh(x) \coloneqq \frac{\partial h}{\partial x}(x)f(x)$. 
% A continuous function $\alpha:\mathbb R_+\rightarrow\mathbb R_+$ is said to belong to class$-\mathcal K$ if it is monotonically increasing with $\alpha(0) = 0$. 
% We use $\mathbb B_\epsilon$ to denote a ball of radius $\epsilon>0$ centered at the origin. 
% Let $d_S(x,y)$ denote the  between points $x,y\in \partial S$ along the boundary of a connected set $S$. 

\section{Problem formulation}\label{sec: math prelim}
% \subsection{Notation and Definition}
%We denote by $\dot x(t) = \frac{\text{d}x(t)}{\dt}$ the time derivative of function $x(t)$ and by $y'(s) = \frac{\text{d}y(s)}{\ds}$ the derivative of $y(s)$ with respect to the stretched time $s$.
% We use a standard notation throughout this paper. Specifically,
% \subsection{Notation}
% Given $x\in \mathbb R^n$ and $y\in \mathbb R^m$, $(x,y) \coloneqq [x^T, y^T]^T$. 
We start with defining a model for the system and the attacker considered in this paper. 
% Then, we present the problem formulation and an outline of our approach for solving it.   

\subsection{System model}
Consider a nonlinear control system $\mathcal S$ given as
\begin{align}\label{eq: actual system}
% \hspace{-40pt}\text{\textit{Actual system}:} &   \hspace{30pt} \dot x = \Phi(t,x,u), \quad x(0) = x_0,
\mathcal S : \begin{cases}\dot x = F(x,u) + d(t,x),\\
x\in \mathcal D, u\in \mathcal U,\end{cases}
\end{align}
where $F:\mathcal D\times\mathcal U\to \mathbb R^n$ is a known function continuous on $:\mathcal D\times\mathcal U$, with $\mathcal D\subset\mathbb R^n$ and $\mathcal U\subset \mathbb R^m$, $d:\mathbb R_+\times\mathbb R^n\rightarrow\mathbb R^n$ is unknown and represents the unmodeled dynamics, $x\in \mathcal D$ is the system state, and $u\in \mathcal U$ is the control input. 
% Since the unmodeled dynamics $d$ is unknown, we work with the following model for the system
% \begin{align}\label{eq: model}
% \dot x = F(x,u). 
% \end{align}
\subsection{Attacker model}
In this paper, we consider attacks on the control input of the system. In particular, we consider an attack where a subset of the components of the control input is compromised. Under such an attack, the system input takes the form:
\begin{align}\label{eq: attack model}
   u = (u_v, u_s),
\end{align}
where $u_v\in \mathcal U_v\subset \mathbb R^{m_v}$ represents the \textit{vulnerable} components of the control input that might be compromised or attacked, and $u_s\in \mathcal U_s\subset \mathbb R^{m_s}$ the \textit{secure} part that cannot be attacked, with $m_v + m_s = m$ and $\mathcal U\coloneqq \mathcal U_v\times \mathcal U_s$.
Under this class of attack, we assume that we know which components of the control input are vulnerable. For example, if the system has four inputs so that $u = \begin{bmatrix}u_1& u_2& u_3& u_4\end{bmatrix}^T$, and $u_1, u_3$ can be attacked, then we assume that this information is known, and $u_v$ is comprised of $u_1$ and $u_3$. We discuss how to address the assumption of which components of the control input are vulnerable in Remark \ref{rem: attack model 1} in Section \ref{sec: num method}. 
% We assume that under this class of attacks, the control input constraint set is given as $\mathcal U \coloneqq \begin{bmatrix}\mathcal U_v\\ \mathcal U_s\end{bmatrix}\subset \mathbb R^m$ $u_v$, so that $u_v\in \mathcal U_v$ and $u_s\in \mathcal U_s$.

% 2) \textit{Additive bias}: 
% Note that the attack model \eqref{eq: attack model} also capture an attack where an additive bias is introduced in the control command. Under such an attack, the system input takes the following form:
% \begin{align}\label{eq: attack model additive attack}
%     u = u_{nom} + \tilde u_v,
% \end{align}
% where $u_{nom} = (u_{nom,a}, u_s)$ is the nominal control input and $\tilde u_v$ is the attack vector with $\tilde u_v = (u_v-u_{nom,a}, 0)$ so that the resulting input is $u = (u_v, u_s)$. 

% Under this class of attack, we assume that the set $\tilde {\mathcal U}_v$ represents a fraction of the input constraint set $\mathcal U$. For example, consider a case of a 1-D system with $\mathcal U = \{u\;| \; |u|\leq u_M\}$ for some $u_M>0$. In this case, $\tilde{ \mathcal U}_v$ can be considered to be of the form $\tilde {\mathcal U}_v = \{u \; |\; |u|\leq pu_M\}$ for some $0<p<1$. This attack model represents the scenario where an attacker can manipulate the actuators up to a fraction of their limits. 

Such attack models have been used in prior work, see e.g., \cite{giraldo2020daria}, and can be implemented in practice by designing the dynamic range of the actuator to preserve its bounds. It can also be implemented in software with the help of a reference monitor~\cite{erlingsson2004inlined} between the controller and the actuator that can check if the desired control inputs satisfy the security policy~\cite{giraldo2020daria}. As discussed in \cite{pasqualetti2013attack}, various prototypical attacks, such as \textit{stealth attacks}, \textit{replay attacks}, and \textit{false-data injection attacks} can be captured by the attack model in \eqref{eq: attack model}. In addition to representing a real-world scenario where system actuators have physical limits, constraining the vulnerable control input $u_v$ in the set $\mathcal U_v$ has several advantages:
\begin{itemize}
    \item[1)] It restricts how much an attacker can change the nominal operation of the system \cite{kafash2018constraining}, and can be implemented physically, so an attacker cannot bypass it.
    \item[2)] It can be utilized to design a detection mechanism, e.g., if $u_v\notin \mathcal U_v$, a flag can be raised signifying that the system is under an attack. Schemes that raise a threshold-based flag are commonly used as detection mechanism \cite{renganathan2020distributionally}. 
    % Ideally, an attacker will try to inject undetectable signals. 
    % and so, it injects an attack input that satisfies $u_v\in \mathcal U_v$. 
    \item[3)] The constraint set $\mathcal U_v$ can be designed appropriately such that the system remains \textit{secured} under attacks, as discussed in Section \ref{sec: num method} (see also \cite{giraldo2020daria,kafash2018constraining}). 
\end{itemize}

\begin{figure}[t]
	\centering
	\includegraphics[width=0.8\columnwidth,clip]{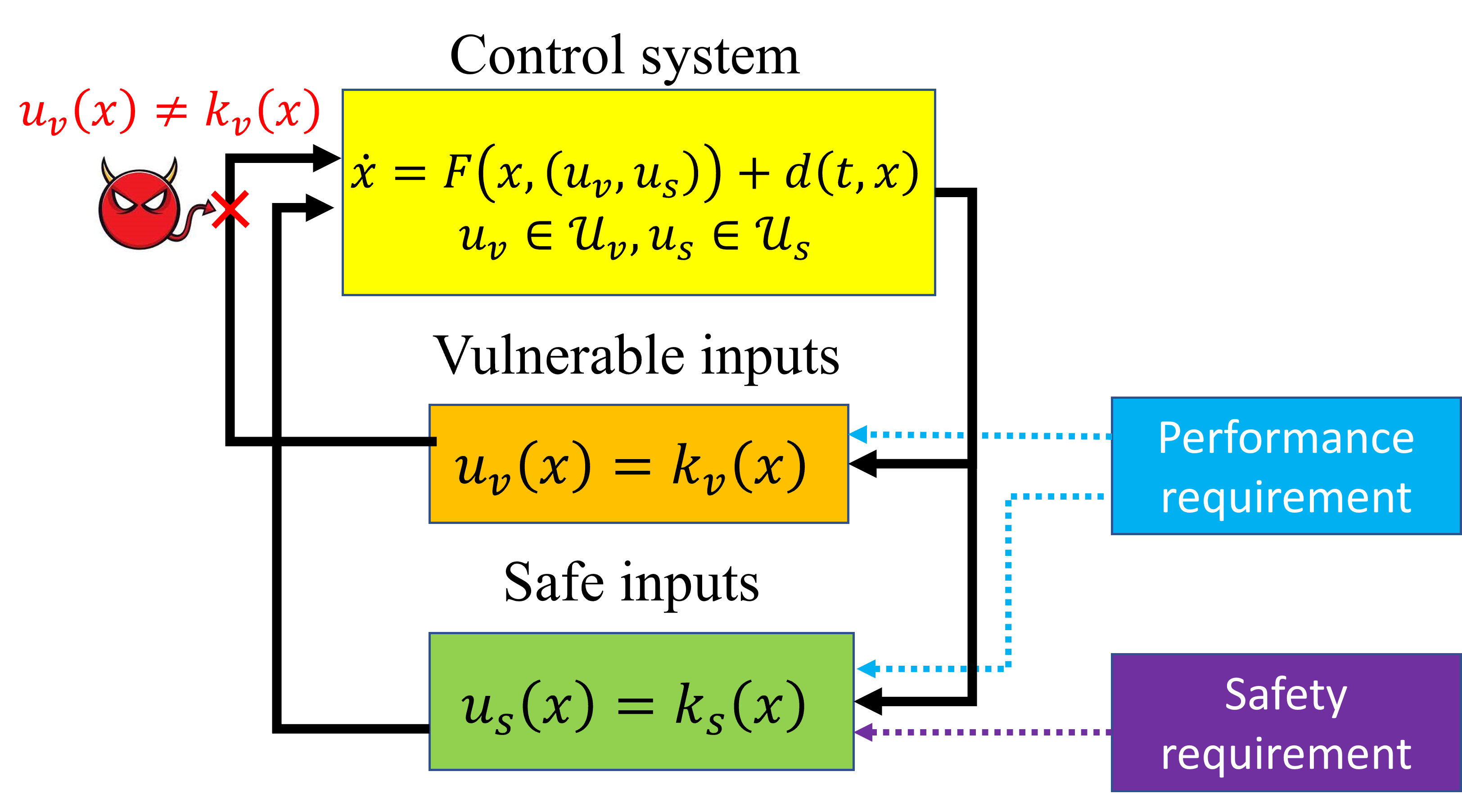}
	\caption{Overview of our approach.}
	\label{fig:overview}
\end{figure}

Now, we present the control design objectives. Consider a non-empty, compact set $S\subset\mathbb R^n$, referred to as safe set, to be rendered forward invariant. We make the following assumption on the unmodeled dynamics $d$ in \eqref{eq: actual system}:
\begin{Assumption}\label{assum: d bound}
There exists $\delta>0$ such that $|d(t, x)|\leq \delta$ for all $t\geq 0$ and $x\in \mathcal D$.
\end{Assumption}
\noindent We consider two properties when designing the control law, an \textit{essential property (safety)}, imposed while designing the secure input $u_s$, and a \textit{desirable property (performance)}, imposed while designing both $u_s$ and $u_v$ (see Figure \ref{fig:overview}). The problem we study in this paper is as follows. 

\begin{Problem}\label{Problem 1}
Given the system in \eqref{eq: actual system} with unmodeled dynamics $d$ that satisfies Assumption \ref{assum: d bound}, a set $S$ and the attack model in \eqref{eq: attack model}, design a feedback law $k_s:\mathbb R^n\rightarrow \mathcal U_s$, and find a set of initial conditions $X_0\subset S$ and the input constraint set $\tilde{\mathcal U}_v\subset\mathcal U_v$, such that for all $x(0)\in X_0$ and $u_v:\mathbb R_+\rightarrow  \tilde{\mathcal U}_v$, the closed-loop trajectories $x:\mathbb R_+\rightarrow\mathbb R^n$ of \eqref{eq: actual system} resulting from using $u_s = k_s(x)$ satisfy  $x(t) \in S$ for all $t\geq 0$.
\end{Problem}

In plain words, we consider the problem of designing a feedback law $u_s$ and compute a set of initial conditions $X_0$ and input constraint set $\tilde{\mathcal U}_v$, such that even under an attack as per the attack model \eqref{eq: attack model}, the system trajectories do not leave the safe set $S$, i.e., the system is secure by design.
% the safety is guaranteed under the attack model \eqref{eq: attack model}. To this end, we propose a method of choosing the input constraint set $\mathcal U_v$ such that the system remains safe for any attack signal $u_v:\mathbb R_+\rightarrow\tilde{\mathcal U}_v$, and thus, the system is secure by design. 
In this work, we assume that the safe set is given as $S \coloneqq \{x\; |\; B(x)\leq 0\}$ where $B:\mathbb R^n\rightarrow\mathbb R$ is a sufficiently smooth user-defined function.

\subsection{Outline of approach}
Given a control system \eqref{eq: actual system}, and an attack model \eqref{eq: attack model}, we first identify a safe set $S\subset \mathbb R^n$ and the vulnerable input $u_v$. Then, our approach to solving Problem \ref{Problem 1} involves the following steps (see Figure \ref{fig:overview 2}):
\begin{itemize}
    \item[1)] \textit{Establish the existence of $X_0$ and $\mathcal U_v$ (Section \ref{sec: safety})}: leverage CBFs to find sufficient conditions to check whether there exist a set of initial conditions $X_0$, input constraint set $\tilde{\mathcal U}_v\subset\mathcal U_v$, and a control input $u_s$ for all $x\in X_0$ that can solve Problem \ref{Problem 1};
    \item[2)] \textit{Numerical method for computation of $X_0$ and $\mathcal U_v$ (Section \ref{sec: num method})}: use conditions in step 1) to formulate a numerical method for computing sets $X_0$ and $\tilde{\mathcal U}_v$;
    \item[3)] \textit{Feedback law synthesis (Section \ref{sec: QP control})}: use the sets $X_0$ and $\tilde{\mathcal U}_v$ from step 2) to design a feedback control law $u_s = k_s(x)$ that solves Problem \ref{Problem 1}.
\end{itemize}

\begin{figure}[t]
	\centering
	\includegraphics[width=\columnwidth,clip]{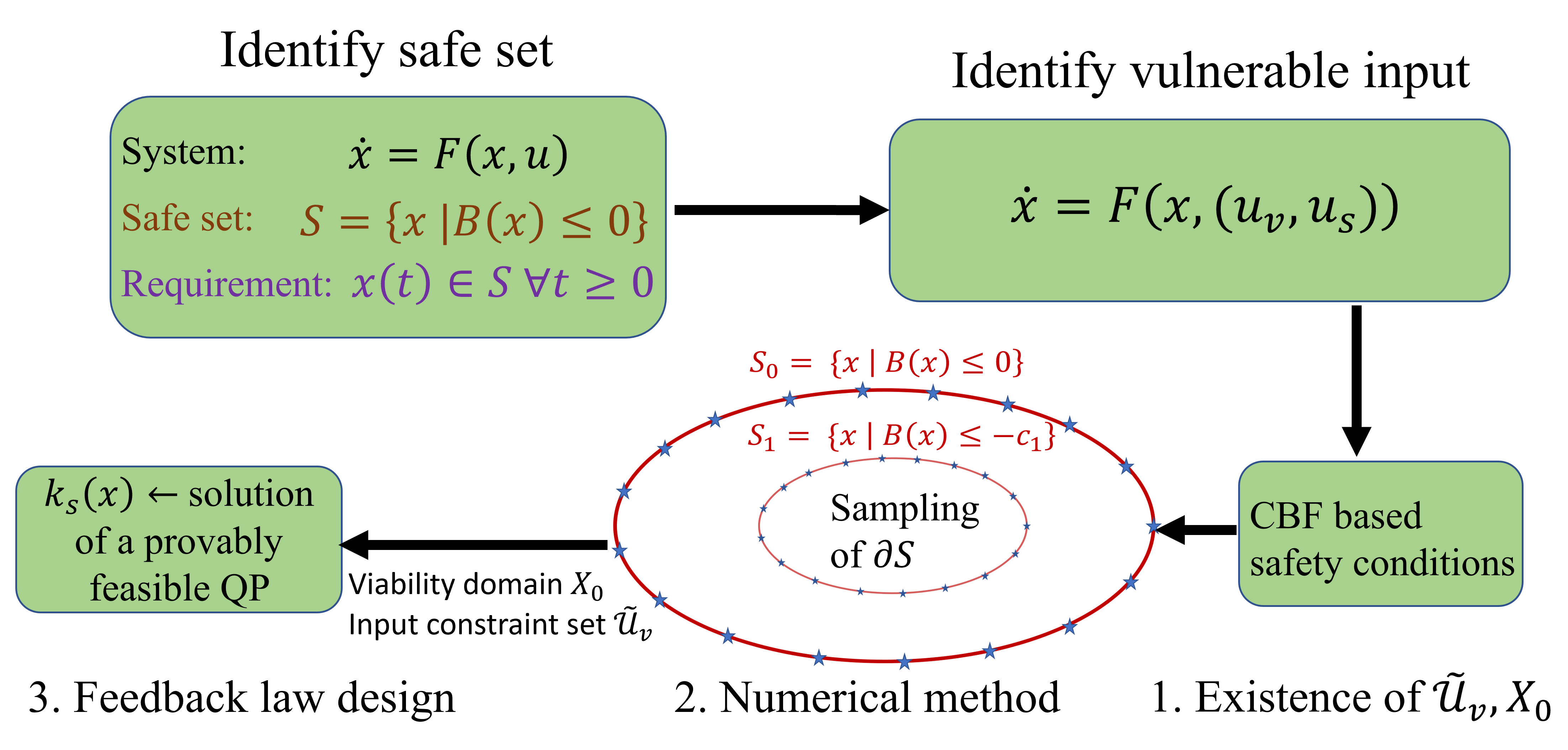}
	\caption{Approach for safe feedback design under attacks.}
	\label{fig:overview 2}
\end{figure}

\subsection{Preliminaries}
Next, we present preliminaries on forward invariance. 
\begin{Definition}
   A set $S\subset\mathbb R^n$ is termed as \textnormal{forward invariant} for system \eqref{eq: actual system} if every solution $x:\mathbb R_+\rightarrow\mathbb R^n$ of \eqref{eq: actual system} satisfies $x(t)\in S$ for all $t\geq 0$ and for all initial conditions $x(0)\in S$. 
\end{Definition}
Next, we review a sufficient condition for guaranteeing forward invariance of a set in the absence of an attack. For the sake of simplicity, in what follows, we assume that every solution of \eqref{eq: actual system} exists and is unique in forward time for all $t\geq 0$ whether or not there is an attack on the system.\footnote{In this work, we assume that even under attack, the solution of the system is unique in forward time. It is possible to study the case when this assumption does not hold using the notion of strong invariance (see \cite{clarke2008nonsmooth}).} 

\begin{Lemma}[\cite{blanchini1999set}]\label{lemma: nec suff safety d 0}
% Let every solution of \eqref{eq: actual system} exist and be unique in forward time. Then, 
Given a continuously differentiable function $B:\mathbb R^n\rightarrow\mathbb R$, the set $S = \{x\; |\; B(x)\leq 0\}$ is forward invariant for \eqref{eq: actual system} with $d\equiv 0$ if the following condition holds:
\begin{equation}\label{eq: safety cond d 0}
    % \inf_{u\in \mathcal U}\left\{L_{f}h(x)+L_{g}h(x)u\right\}\leq \alpha(-h(x)),
    \inf_{u\in \mathcal U}L_{F}B(x,u)\leq 0 \quad \forall x\in \partial S.
\end{equation}
% Furthermore, if $h(x(0))<0$, then for any $T\geq 0$, $h(x(t))<0$ for all $0\leq t\leq T$. 
\end{Lemma}

Satisfaction of \eqref{eq: safety cond d 0} means the set $S$ is termed as a \textit{viability} domain for system \eqref{eq: actual system} with $d\equiv 0$. Using Lemma \ref{lemma: nec suff safety d 0}, and following the notion of robust CBF in \cite{garg2021robust}, we can state the following result guaranteeing forward invariance in the presence of disturbance $d$.

\begin{Lemma}[\cite{garg2021robust}]\label{lemma: nec suff safety}
% Let every solution of \eqref{eq: actual system} exist and be unique in forward time. Then, 
Given a continuously differentiable function $B:\mathbb R^n\rightarrow\mathbb R$, the set $S = \{x\; |\; B(x)\leq 0\}$ is forward invariant for \eqref{eq: actual system} under $d$ satisfying Assumption \ref{assum: d bound} if the following condition holds:
\begin{equation}\label{eq: safety cond}
    % \inf_{u\in \mathcal U}\left\{L_{f}h(x)+L_{g}h(x)u\right\}\leq \alpha(-h(x)),
    \inf_{u\in \mathcal U}L_{F}B(x,u)\leq -l_B\delta \quad \forall x\in \partial S,
\end{equation}
where $l_B$ is the Lipschitz constant of the function $B$.
% Furthermore, if $h(x(0))<0$, then for any $T\geq 0$, $h(x(t))<0$ for all $0\leq t\leq T$. 
\end{Lemma}

\section{Sufficient conditions for safety}\label{sec: safety}
In this section, we present sufficient conditions that guarantee the security of the system model \eqref{eq: actual system} against attacks on the input. We say that the system \eqref{eq: actual system} is \textit{secure with respect to the safety property} for a set $S$ if for all initial conditions $x(0)\in S$,  $x(t)\in S$ for all $t\geq 0$, $u_v\in \mathcal U_v$ and $d$ satisfying Assumption \ref{assum: d bound}. 
% The following lemma provides necessary and sufficient conditions for a system to be not secured with respect to the safety property $P_{safe}$. 
% \begin{Lemma}
% The system is locally vulnerable at point $x_a\in \partial S$ if and only if the following holds:
% \begin{align}\label{eq: attack sufficiency}
%     % \sup_{u_v\in \mathcal U_v}\inf_{u_s\in \mathcal U_s}\{L_fB(x_a) + L_{g_s}B(x_a)u_s +L_{g_v}h_{S}(x_a)u_v\}>\epsilon.
%     \sup_{u_v\in \mathcal U_v}\inf_{u_s\in \mathcal U_s} L_F h_{S}(x_a,u_v,u_s)>0.
% \end{align}
% It is globally vulnerable if the above holds for all $x_a\in \partial S$.  
% \end{Lemma}
% The proof follows immediately from Definition \ref{def: safety secured}. 
% Conversely, 
Given $B$ and $F$, define $H:\mathbb R^n\times\mathbb R^{m_v}\rightarrow\mathbb R$:
\begin{align}\label{eq: H def}
    % H(x,u_v) \coloneqq L_fB(x) + L_{g_s}B(x)u_s(x) +L_{g_v}h_{S}(x)u_v.
    H(x,u_v) \coloneqq \inf_{u_s\in \mathcal U_s}L_F B(x,(u_v,u_s)).
\end{align}
It is not necessary that the zero sublevel set $S$ of the function $B$ is a viability domain for system \eqref{eq: actual system}. Any non-empty sublevel set $S_c \coloneqq \{x\; |\; B(x)\leq -c\}$, where $c\geq 0$, being a viability domain is sufficient for safety of the system. Note that the set $S_c$ is non-empty for $0\leq c\leq -\min\limits_{x\in S}B(x)$. Define $c_M\in \mathbb R$ as
\begin{align}\label{eq: cM}
c_M \coloneqq -\min\limits_{x\in S}B(x),
\end{align}
so that the set of feasible values for $c$ is given as $[0, c_M]$.\footnote{Note that compactness of the set $S$ guarantees existence of $c_M$.} The following result provides sufficient conditions for a system to be secured with respect to the safety property. 
% $P_{safe}$:
\begin{Proposition}\label{lemma: sufficient conditions non-attack}
Suppose there exist $c\in [0, c_M]$ and nonempty $\tilde{\mathcal U}_v\subset\mathcal U_v$ such that
\begin{align}\label{eq: invulnerable cond}
    \sup_{u_v\in \tilde{\mathcal U}_v}H(x,u_v) \leq -l_B\delta \quad \forall x\in \partial S_c,
\end{align}
and the system solutions are uniquely defined in forward time for all $x(0)\in S_c$. Then, for each $d$ satisfying Assumption \ref{assum: d bound}, system \eqref{eq: actual system} is secured with respect to the safety property for the set $S_c$.
% for all $\tau\geq 0$, i.e., $[\mathcal A_\tau]\mathcal S\models P_{safe, S_c}$ for all $\tau\geq 0$.  
% one of the following conditions hold:
% \begin{itemize}
%     \item[1)] $L_{g_v}B(x) = 0$ and $\inf_{u_s\in \mathcal U_s}\{L_fB(x) + L_{g_s}B(x)u_s\}\leq 0$ for all $x\in \partial S$;
%     \item[2)] $\sup_{u_v\in \mathcal U_v}\inf_{u_s\in \mathcal U_s}\{L_fB(x) + L_{g_s}B(x)u_s +L_{g_v}h_{S}(x)u_v\}\leq 0$ for all $x\in \partial S$; 
% \end{itemize}
\end{Proposition}
\begin{proof}
% Note that Definition \ref{def: safety secured} requires that for a point $x_a$ to be vulnerable, the strict inequality \eqref{eq: attack def} must be satisfied for all points in the neighborhood of the point. Thus, excluding a measure zero set, 
Note that 
\begin{align*}
    \inf_{u_s\in\mathcal U_s\atop u_v\in \tilde{\mathcal U}_v}L_F(x,(u_v, u_s)) = \inf_{u_v\in \tilde{\mathcal U}_v}H(x,u_v)\leq \sup_{u_v\in \tilde{\mathcal U}_v}H(x,u_v).
\end{align*}
Thus, from  \eqref{eq: invulnerable cond}, it follows that \eqref{eq: safety cond} holds. Thus, per Lemma \ref{lemma: nec suff safety}, the set $S_c$ is forward invariant and it holds that the system \eqref{eq: actual system} is secured with respect to the safety property for set $S_c$.  
% and it holds that $[\mathcal A_\tau]\mathcal S\models P_{safe, S_c}$ for all $\tau\geq 0$. 
\end{proof}

{\noindent Note that satisfaction of the conditions in Proposition \ref{lemma: sufficient conditions non-attack} implies that for all $x\in \partial S_c$ and $u_v\in \bar{ \mathcal U}_v$,  there exists an input $u_s\in \mathcal U_s$ such that the inequality $L_FB(x,(u_v, u_s))\leq -l_B\delta$ holds. This, in turn, implies that the set $S_c$ is a viability domain for system \eqref{eq: actual system}. Condition \eqref{eq: invulnerable cond} requires checking the inequality $\sup\limits_{u_v\in \tilde{\mathcal U}_v}H(x,u_v)\leq -l_B\delta$ for all points on the boundary of the set $S_c$. Such conditions are commonly used in the literature for control synthesis, assuming that the viability domain is known. However, it is not an easy task to compute a viability domain in practice for a general class of nonlinear systems \textit{a priori}. In the next section, we present a computationally tractable method where we show that checking a modification of the inequality in \eqref{eq: invulnerable cond} on a set of sampling points on the boundary is sufficient. }

\section{Viability domain under bounded inputs}\label{sec: num method}
In this section, we present numerical algorithms to assess whether given system \eqref{eq: actual system} and the function $B$, there exist $c$ and an input constraint set $\tilde{\mathcal U}_v$ such that condition \eqref{eq: invulnerable cond} holds. First, we present a sampling-based method for evaluating whether the condition \eqref{eq: invulnerable cond} holds by checking a modified inequality at a finite set of sampling points. Then, we propose an iterative method to compute $c$ and the set $\tilde{\mathcal U}_v$.
% to facilitate the sampling-based method.}

% \begin{figure}[t]
% 	\centering
% 	\includegraphics[width=1\columnwidth,clip]{fig_sample_safety.png}
% 	\caption{safe system trajectories under an attack on the inputs. The solid line shows the system trajectories before the attack and the dashed line, after the attack starts at $t = \tau$. The system trajectories without any attack after $t = \tau$ is shown in lighter solid line. }
% 	\label{fig:non-attack safety}
% \end{figure}
\subsection{Viability domain using sampling data}
% If we assume sufficient regularity on the functions $B$ and $F$, numerical methods can be devised to check a modified condition at sampling points while still guaranteeing that the required condition \eqref{eq: invulnerable cond} holds for all $x$ on the boundary of a sublevel set of the function $B$. 
We start by making the following assumption on the regularity of the function $H$ defined in \eqref{eq: H def}.
\begin{Assumption}\label{assum: LC F del B}
The function $\sup_{u_v\in \tilde{\mathcal U}_v}H(\cdot,u_v)$
% $\phi:\mathbb R^n\rightarrow\mathbb R$ defined as
% \begin{align}
%     \phi(x) = \sup_{u_v\in \tilde{\mathcal U}_v}H(x,u_v),
% \end{align}
is Lipschitz continuous on $S$ with constant $l_H>0$. 
\end{Assumption}

\noindent 
First, to illustrate the method, we consider the 3-D case, i.e., when $x\in \mathbb R^3$. 
% \textit{Triangulation} of surface in $\mathbb R^3$ is defined by the authors in \cite{hartsfield1991clean} as follows.
% \begin{Definition}\label{def: triangulation}
%   A polyhedron $P$ with its vertices on a closed surface $S\subset \mathbb R^3$ is called a \textnormal{triangulation} of $S$, or said to \textnormal{triangulate} the surface $S$, if each face of $P$ is a triangle with three distinct vertices, and the intersection of any two distinct triangles is either empty, a single vertex, or a single edge. 
% \end{Definition}
If the compact set $S\subset \mathbb R^3$ is diffeomorphic to a unit sphere in $\mathbb R^3$, then the sampling points on the boundary of the unit sphere can be used to obtain the points on the boundary of $S_c$. Thus, without loss of generality, we can study the case when $S\subset \mathbb R^3$ is a unit sphere with center $x_o\in \mathbb R^3$. Let $\{x_i\}_{\mathcal I}$, with each $x_i\in \partial S_c$, denote the set of $N_p$ sampling data points on the boundary of the sublevel set $S_c$ for a given $c\in [0, c_M]$ with $c_M$ defined in \eqref{eq: cM} and $\mathcal I \coloneqq \{1, 2, \dots, N_p\}$. The sampling points $\{x_i\}_{\mathcal I}$ are such that they constitute a polyhedron $P_{\mathcal I}$ with $N_f>0$ triangular faces, $T_1, T_2, \dots, T_{N_f}$, such that $P_{\mathcal I}$ \textit{triangulates} the boundary $\partial S_c$, i.e., the intersection of any two distinct triangles is either empty, a single vertex, or a single edge. Figure \ref{fig:3d sampling} shows an example of triangulation of a unit sphere in $\mathbb R^3$. Interested readers on algorithms and details on triangulation are referred to \cite{oudot2003provably}, and the references therein.

Note that a tetrahedron is the \textit{minimal} triangulation (i.e., a triangulation with minimum number of triangular faces) for a unit sphere. Using geometric arguments, it is easy to show that the minimum possible value of the maximum of the inter-vertex distances for a tetrahedron inscribed in a unit sphere is $\sqrt{3}$. The corresponding arc-length along the boundary of the unit sphere (denote as $d_a$) is $2\sin^{-1}\sqrt{\frac{3}{4}}$. It follows that if $d_a\leq 2\sin^{-1}\sqrt{\frac{3}{4}}$, then there must be at least $N_p = 4$ points in the polyhedron. Finally, with $0\leq r_c\leq 1$ being the radius of the sphere $S_c$,\footnote{If the set $S_c$ is defined as $S_c = \{x\; |\; |x|^2-1\leq -c\}$, then $r_c = \sqrt{1-c}$.} the corresponding arc-length for $S_c$ is
\begin{align}\label{eq: def dM}
    d_M \coloneqq 2r_c\sin^{-1}\sqrt{\frac{3}{4}}.
\end{align}

\noindent Now, to ensure that they are enough sampling points, the following conditions can be imposed on $\{x_i\}_{\mathcal I}$ for a given $c\in [0,  c_M]$ and $d_a\in \begin{bmatrix}0, d_M\end{bmatrix}$ 
\begin{itemize}
    \item For each $x\in \partial S_c$, there exists a triangular face $T_j$ with vertices $x_{j_1}, x_{j_2}, x_{j_3}\in \{x_i\}_{\mathcal I}$, of the polyhedron $P_{\mathcal I}$ generated by $\{x_i\}_{\mathcal I}$, such that $x_o+\theta (x-x_o)\in T_j$ for some $0\leq \theta\leq 1$; and 
    \item The following holds:
\begin{align}\label{eq: xj xk dist cond 3d}
  \max_{l\neq m \atop l, m = 1, 2, 3}d_{S_c}(x_{j_l}, x_{j_m})\leq d_a,
\end{align}
where $d_{S_c}(x,y)$ denotes the shortest  arc-length between the points $x,y\in \partial S_c$.
\end{itemize}

In plain words, the above conditions require for each point $x\in \partial S_c$, the line joining the center $x_o$ and $x$ intersects a triangular face of the polyhedron such that the distance along the boundary $\partial S_c$ between the vertices of this face is bounded by $d_a$. Note that smaller $d_a$ requires larger number of sampling points $N_p$. 
% for satisfaction of Assumption \ref{assum: max dist tring}. 
% Finally, it is easy to verify that $|x-y|\leq d_{S_c}(x_j,x_k)$ for all $j, k\in \mathcal I$, i.e., the Euclidean distance between two sampling points is less or equal than their distance computed along the boundary $\partial S_c$. 
% We can now state the following result. 
% We assume that we have enough sampling points that are sufficient to cover the surface $\partial S_c$. 
%, i.e., the distance between the consecutive sampling points along the surface $\partial S_c$ is no greater than $\frac{\gamma}{l_H}$. 
\begin{figure}[t]
	\centering
	\includegraphics[width=0.3\columnwidth,clip]{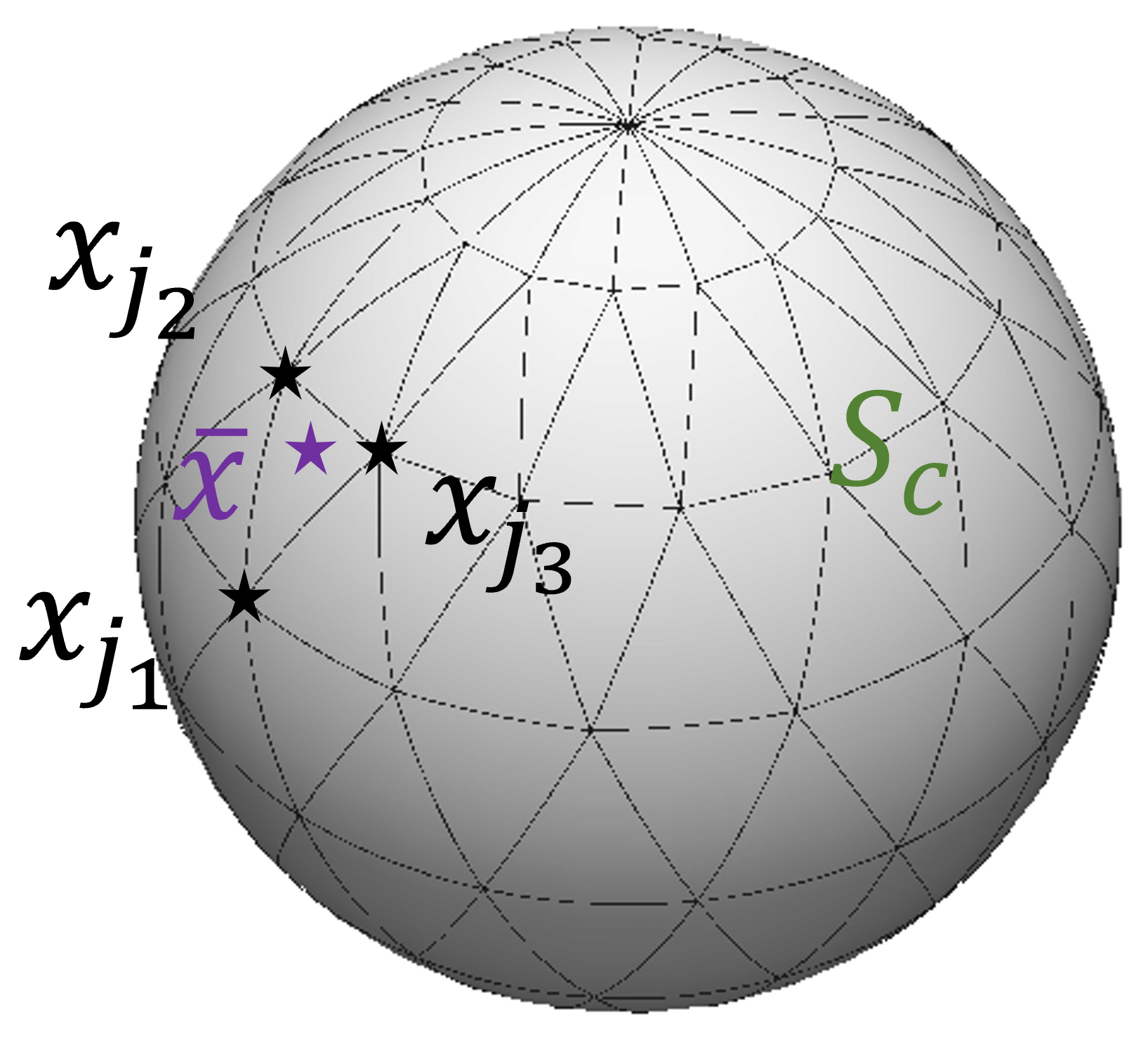}
	\caption{3-D case: Triangulating sampling of the boundary $\partial S_c$.}
	\label{fig:3d sampling}
\end{figure}
% \begin{Lemma}\label{lemma: sufficient cond 3D}
Now, we show that 
% the function $H$ defined in \eqref{eq: H def} satisfies Assumption \ref{assum: LC F del B}. Given $c\in [0, c_M]$, $d_a\in \begin{bmatrix}0,  d_M\end{bmatrix}$, and the sampling points $\{x_i\}_{\mathcal I}$
% such that the resulting polyhedron $P_{\mathcal I}$ triangulates the boundary $\partial S_c$, 
% satisfying Assumption \ref{assum: max dist tring}, if 
if the following holds
\begin{align}\label{eq: sup H cond 3D}
    \sup_{u_v\in \tilde{\mathcal U}_v}H(x_i,u_v)\leq -l_Hd_a-l_B\delta \quad \forall i\in \mathcal I,
\end{align}
where $l_B$ is the Lipschitz constant for $B$ and $\delta, l_H$ are as defined in Assumptions \ref{assum: d bound} and \ref{assum: LC F del B}, respectively,
then, \eqref{eq: invulnerable cond} holds.
% the system is secured with respect to the safety property $P_{safe,S_c}$.  
% then, \eqref{eq: invulnerable cond}
% condition 2) of Proposition \ref{lemma: sufficient conditions non-attack}
% holds for $\tilde h(x) = B(x) + c$. 
% \end{Lemma}
% \begin{proof}
% Suppose there exists $c\in [0, c_M]$ such that \eqref{eq: sup H cond 3D} holds, and a corresponding $d\in \begin{bmatrix}0,  2\sin^{-1}\sqrt{\frac{3}{4r_c}}\end{bmatrix}$ such that Assumption \ref{assum: max dist tring} holds.
With $c\in [0, c_M]$, the set $S_c$ is non-empty, and with $d_a\in \begin{bmatrix}0,  d_M\end{bmatrix}$, there exist sufficient points $N_p$ to have a polyhedron that can triangulate the boundary $\partial S_c$. Under the conditions imposed on $\{x_i\}_{\mathcal I}$, for every $\bar x\in \partial S_c$, there exists a triangular face $T_j$ with coordinates $x_{j_1}, x_{j_2}, x_{j_3}\in \{x_i\}_{\mathcal I}$   satisfying \eqref{eq: xj xk dist cond 3d} and $0\leq \theta\leq 1$ such that $x_o+\theta(\bar x-x_o)\in T_j$. Using Lipschitz continuity of $\sup_{u_v}H(\cdot, u_v)$ per Assumption \ref{assum: LC F del B}, it holds that $\sup_{u_v\in \tilde{\mathcal U}_v}H(\bar x,u_v) \leq  \sup_{u_v\in \tilde{\mathcal U}_v}H(x,u_v) + l_H|\bar x-x|$,
% the following holds 
for all $x, \bar x\in \partial S_c$.
% \begin{align*}
% \end{align*}
For $x = x_i$, $i\in \mathcal I$, using \eqref{eq: sup H cond 3D} and the fact that $|x-y|\leq d_{S_c}(x,y)$ for all $x, y\in \partial S_c$, we obtain that
\begin{align*}
    \sup_{u_v\in \tilde{\mathcal U}_v}H(\bar x,u_v)\leq & -l_Hd_a-l_B\delta + l_Hd_{S_c}(\bar x, x_i), \; \forall i\in \mathcal I.
\end{align*}
% where the last inequality follows from \eqref{eq: sup H cond 3D} and 
% Using the above inequality for $x_i = x_{j_1}, x_{j_2}$ and $x_{j_3}$, we obtain that for all $\bar x\in \partial S_c$, 
% \begin{align*}
%     \sup_{u_v\in \tilde{\mathcal U}_v}H(\bar x,u_v) \leq & -l_Hd_a -l_B\delta+ \frac{l_H}{3}\sum_{k = 1}^3d_{S_c}(\bar x, x_{j_k}).
% \end{align*}
Since the projection of $\bar x$ lies in the triangular face $T_j$, it holds that $d_{S_c}(\bar x, x_{j_k})\leq d_{S_c}(x_{j_l}, x_{j_k})$ for $l\neq k$, $l, k \in \{1, 2, 3\}$. Using this and \eqref{eq: xj xk dist cond 3d}, we obtain
% {\small
% \begin{align*}
%     \sup_{u_v\in \tilde{\mathcal U}_v}H(\bar x,u_v) \leq & -l_Hd_a -l_B\delta+ \frac{l_H}{3}\left(3\max_{l \neq m, \atop l, m = 1, 2, 3}d_{S_c}(x_{j_l}, x_{j_m})\right).
% \end{align*}}\normalsize
% Moreover, using \eqref{eq: xj xk dist cond 3d}, it follows that
\begin{align*}
    \sup_{u_v\in \tilde{\mathcal U}_v}H(\bar x,u_v) \leq -l_Hd_a -l_B\delta+l_Hd_a = -l_B\delta, \; \forall \bar x\in \partial S_c.
\end{align*}
% for all $\bar x\in \partial S_c$.
% in the projection of the triangle $T_j$ on $\partial S_c$. Under Assumption \ref{assum: max dist tring}, using the similar arguments for sampling points on all of the triangular faces on the boundary, it can be concluded that $\sup_{u_v\in \tilde{\mathcal U}_v}H(\bar x,u_v)\leq -l_B\delta$ for all $\bar x\in \partial S$.
% \end{proof}
Thus, checking the inequality \eqref{eq: sup H cond 3D} at a finite number of points is a computationally tractable method for assessing whether \eqref{eq: invulnerable cond} holds for a given $c$ and $\tilde{\mathcal U}_v$. Note that for a given $F, B, \tilde{\mathcal U}_v$ and $\delta$, a smaller value of $d_a$ implies that the right-hand side of \eqref{eq: sup H cond 3D} is less negative, thus, making it easier to satisfy the inequality. At the same time, due to \eqref{eq: xj xk dist cond 3d}, a smaller value of $d_a$ requires more sampling points $N_p$, and hence, checking the inequality at more points. Thus, there is a trade-off between the ease of satisfaction of \eqref{eq: sup H cond 3D} and the number of points at which the inequality should be checked.

% \begin{Remark}
% While we present the specific cases of 2D and 3D for the ease of visualization, the presented method can be readily extended to $n-$dimension, i.e., the case when $S_c\in \mathbb R^n$, using a sampling for the corresponding $(n-1)$ dimensional manifold $\partial S_c$. In particular, using the notion of maximum distance $d$ between the consecutive points, noting that a convex combination of affinely independent $n$ points can cover a portion of the surface $\partial S_c$ with non-zero Legesgue measure and sampling the surface $\partial S_c$ using a $(n-1)-$simplex (see e.g., \cite{karimipour2010simplex}), Lemma \ref{lemma: sufficient cond 3D} can be used for a general $n-$dimensional system.
% \end{Remark}
% In plain words, if the distance between the consecutive sampling-points ${x_i}$ is no more than $\frac{\gamma}{l_H}$, then satisfaction of \eqref{eq: sup H cond} implies that the system is secured with respect to the set $S_c = \{x\; |\; B(x)\leq -c\}$.
% \subsubsection{General case of $S\subset \mathbb R^n$} 
The above arguments can be generalized to the $n-$dimensional case. Using the sampling approach in \cite{leopardi2009diameter} for a unit sphere in $n-$dimension, combined with Delaunay Triangulation of the sampling points (see e.g., \cite{de1997computational}), an $(n-1)-$dimensional \textit{simplex} can be obtained. If the compact set $S\subset \mathbb R^n$ is diffeomorphic to a unit $(n-1)-$sphere, then sampling points on the boundary of $S$ can be obtained using the sampling points for the $(n-1)-$unit sphere. Thus, we study the case when the set $S$ is an $(n-1)-$unit sphere.
% make the following assumption.
% \begin{Assumption}\label{assum: diffe sphere}
% \end{Assumption}
% Under the above assumption, the diffeomorphism can be used to obtain  

Let $\{x_i\}_{\mathcal I}$, with each $x_i\in \partial S_c$, denote the set of $N_p$ sampling data points on the boundary of the sublevel set $S_c$ for a given $c\in [0, c_M]$ with $c_M$ defined in \eqref{eq: cM} and $\mathcal I \coloneqq \{1, 2, \dots, N_p\}$. The sampling poins $\{x_i\}_{\mathcal I}$ constitute a simplex $\mathcal S_{\mathcal I}$ with $N_f>0$ faces, $\mathcal X_1, \mathcal X_2, \dots, \mathcal X_{N_f}$. 
% The convex-hull of a $(n-1)$-simplex covers a subset of $n-$dimensional object with non-zero Lebesgue measure. Thus, using the similar ideas employed for triangulating a surface in 3D, one case uses $(n-1)-$simplex to sample a boundary of $n$-dimensional set $S$.  
For a unit sphere in $\mathbb R^n$, the minimum number of points in the simplex is $(n+1)$, and the minimum possible value of the maximum of the lengths of its edges is $\sqrt{\frac{2(n+1)}{n}}$. The length, denoted as $d_a$, of the corresponding arc-length on the boundary $\partial S_c$ is $2r_c\sin^{-1}\sqrt{\frac{(n+1)}{2n}}$, where $0\leq r_c\leq 1$ is the radius of the sphere $S_c$. Thus, with $d_a\leq 2r_c\sin^{-1}\sqrt{\frac{(n+1)}{2n}}$, there must be at least $(n+1)$ points in the simplex. For the sake of brevity, define $d_{M,n}\coloneqq 2r_c\sin^{-1}\sqrt{\frac{(n+1)}{2n}}$. We make the following assumption on the sampling points $\{x_i\}_{\mathcal I}$.  
\begin{Assumption}\label{assum: max dist tring nD}
Given $c\in [0,  c_M]$, the sampling points $\{x_i\}_{\mathcal I}$ and $d_a\in \begin{bmatrix}0,  d_M\end{bmatrix}$, for each $x\in \partial S_c$, there exists a face $\mathcal X_j$ with vertices $\{x_{j_1}, x_{j_2}, \dots, x_{j_n}\}\in \{x_i\}_{\mathcal I}$, where $j\in \{1, 2, \dots, N_f\}$, of the simplex $\mathcal S_{\mathcal I}$ generated by $\{x_i\}_{\mathcal I}$, such that $x_o+\theta (x-x_o)\in \mathcal X_j$ for some $0\leq \theta\leq 1$, and the following holds:
\begin{align}\label{eq: xj xk dist cond nd}
   \max_{l\neq m \atop l, m = 1, 2, \dots, n}d_{S_c}(x_{j_l}, x_{j_m})\leq d_a,
\end{align}
where $d_{S_c}(x,y)$ denotes the shortest arc-length between the points $x,y\in \partial S_c$.
\end{Assumption}

% Assuming that the maximum distance between any pair of vertices of this simplex satisfies \eqref{eq: xj xk dist cond 3d}, 
% The following result can now be stated. 
We have the following result when $S$ is $(n-1)-$unit sphere. 

\begin{Theorem}\label{thm: sufficient sup H cond}
% Consider a unit sphere $S\subset \mathbb R^n$.
Suppose that the function $H$ defined in \eqref{eq: H def} satisfies Assumption \ref{assum: LC F del B}. Given $c\in [0, c_M]$, $d_a\in \begin{bmatrix}0,  d_{M,n}\end{bmatrix}$, and the sampling points $\{x_i\}_{\mathcal I}$,
% such that the resulting polyhedron $P_{\mathcal I}$ triangulates the boundary $\partial S_c$,
if Assumption \ref{assum: max dist tring nD} and \eqref{eq: sup H cond 3D} hold, 
% Given $d>0$, the sampling points $\{x_i\}_{\mathcal I}$ on the boundary $\partial S_c$ satisfying Assumption \ref{assum: max dist tring nD}, if there exists $0\leq c\leq c_M$ such that \eqref{eq: sup H cond 3D} 
% and Assumption \ref{assum: max dist tring} holds, 
then, \eqref{eq: invulnerable cond} holds.
\end{Theorem}
\begin{proof}
With $c\in [0, c_M]$, the set $S_c$ is non-empty, and with $d_a\in \begin{bmatrix}0,  d_M\end{bmatrix}$, there exist sufficient points $N_p$ to have a simplex. Now, consider any point $\bar x\in \partial S_c$. Under Assumption \ref{assum: max dist tring nD}, for every $\bar x\in \partial S_c$, there exists a face $\mathcal X_j$ of the simplex $\mathcal S_{\mathcal I}$, such that the line joining the center of the sphere $S_c$ and the point $\bar x$ lies on this face. 
% Consider the $n$ vertices $\{x_{j_1},x_{j_2}, \dots, x_{j_n}\}$ of an arbitrary simplex $\mathcal X_j$ used for sampling the boundary $\partial S_c$ and consider any point $\bar x\in \mathbb R^n$ in the projection of $\mathcal X_j$ on $\partial S_c$. 
Using Lipschitz continuity of $\sup_{u_v}H(\cdot, u_v)$ under Assumption \ref{assum: LC F del B} and \eqref{eq: sup H cond 3D}, it holds that
\begin{align*}
    \sup_{u_v\in \tilde{\mathcal U}_v}H(\bar x,u_v) \leq & \sup_{u_v\in \mathcal U_v}H(x,u_v) + l_H|\bar x-x|\\
    \leq & -l_Hd_a-l_B\delta + l_H|\bar x-x|,
\end{align*}
for all $x, \bar x\in \partial S_c$. 
% Using the above inequality for the $n$ vertices $x_{j_1}, x_{j_2}, \dots, x_{j_n}$, of $\mathcal X_j$, we obtain that
% \begin{align*}
%      \sup_{u_v\in \tilde{\mathcal U}_v}H(\bar x,u_v) \leq & -l_Hd_a-l_B\delta + \frac{1}{n}l_H\sum_{i= 1}^n |\bar x-x_{j_i}|\\
%      \leq &  -l_Hd_a-l_B\delta + \frac{1}{n}l_H( n \max_{i\in \{1, \dots, n\}}|\bar x-x_{j_i}|)
% \end{align*}
% Finally, 
Using the inequality for $x = x_{j_i}$, $i\in \{1, \dots, n\}$ and the fact that $|\bar x-x_{j_i}|\leq d_{S_c}(\bar x,x_{j_i})\leq d_{S_c}(x_{j_k},x_{j_i})$ for any $k\neq i$, $k\in \{1, \dots, n\}$ and \eqref{eq: xj xk dist cond nd}, we obtain that
\begin{align*}
     \sup_{u_v\in \tilde{\mathcal U}_v}H(\bar x,u_v)
     \leq & -l_Hd_a-l_B\delta + l_H d_a =  -l_B\delta
\end{align*}
for all $\bar x\in \partial S_c$, which completes the proof.
\end{proof}

Note that there are three set of parameters that can facilitate satisfaction of \eqref{eq: sup H cond 3D} in the following manner:
\begin{itemize}
    \item Set $\tilde{\mathcal U}_v$: smaller $\tilde{\mathcal U}_v$ makes it easier to satisfy \eqref{eq: sup H cond 3D};
    \item Parameter $c$: larger value of $c$ results in smaller values of $d_M$, thus, reducing the right-hand side of \eqref{eq: sup H cond 3D}, and making it easier to satisfy it; and
    \item Number of sampling points $N_p$: larger $N_p$ results in smaller value of $d_{M,n}$.
\end{itemize}

Based on these observations, an iterative algorithm can be formulated to check whether there exists a feasible $c$ and a non-empty set $\tilde{\mathcal U}_v$, such that \eqref{eq: sup H cond 3D} holds. 

\subsection{Iterative algorithm}
% \subsubsection{Computation of $c$}: 
We formulate our algorithm with the following steps:
\begin{itemize}
    \item[1)] For a given value of $0\leq c\leq c_M$, $\tilde{\mathcal U}_v$ and number of sampling points $N_p$, sample $\{x_i\}_{\mathcal I}$ from the set $\partial S_c$ and check if \eqref{eq: sup H cond 3D} holds for all the sampling points;
    \item[2)] Shrink $\tilde{\mathcal U}_v$,  increase $c$ and repeat steps 1)-2) until the condition \eqref{eq: sup H cond 3D} is satisfied for all the sampling points, or there does not exist a $c$ and a non-empty set $\tilde{\mathcal U}_v$;
    \item[3)] Increase $N_p$ and repeat steps 1)-3) until \eqref{eq: sup H cond 3D} holds or the maximum value ($N_{max}$) of $N_p$ is reached.
\end{itemize}
Using these steps, we propose Algorithm \ref{algo: iter U c comb} which returns a feasible $c$ and a set $\tilde{\mathcal U}_v$ such that safety is guaranteed for all $x\in S_c$ and $u_v\in \tilde{\mathcal U}_v$. In other words, this algorithm can compute the set of initial conditions $S_c$, and the set of \textit{tolerable} attacked inputs via $\tilde{\mathcal U}_v$ such that the system can satisfy the safety property under attacks. The order in which the parameters $c, \tilde{\mathcal U}_v$, and $N_p$ are tuned can be changed, which can potentially change the output of the algorithm. 
% In brief, Algorithm \ref{algo: iter U c comb}, for a given $c$, shrinks the set $\mathcal U_v$ until the condition \eqref{eq: sup H cond 3D} is satisfied or the set $\mathcal U_v$ becomes empty. If the set $\mathcal U_v$ becomes empty, it is re-initialized, and the value of $c$ is increased, or equivalently, the set $S_c$ is shrunk, until the condition \eqref{eq: sup H cond 3D} is satisfied. This process is repeated until the set $S_c$ becomes empty, i.e., $c = c_M$. 

\begin{algorithm}[h]\label{algo: iter U c comb}
\SetAlgoLined
\KwData{$f,g_v, g_s,\mathcal U_v, \mathcal U_s, B, d_a, \varepsilon_1, \varepsilon_2, \delta, N_{max}, N_{c0}$}
% \KwResult{Input constraint set $\tilde{\mathcal U}_v$ and $c$}
\textbf{Initialize: }$\tilde{\mathcal U}_v = \mathcal U_v, c = 0, N_p = N_{c0}$\;
\While{$N_p< N_{max}$}
{
 \While{$c\leq c_M$}{
 Sample $\{x_i\}_{\mathcal I}$ from $\{B(x)\leq -c\}$\;
 \While{$\tilde{\mathcal U}_v\neq \emptyset$}{
  \lIf{\footnotesize{$\{i\in \mathcal I\; |\; H(x_i, u_v)\!>\!-l_Hd_a\!+\!l_B\delta\} \neq \emptyset $}\label{line: h s U}}{
  
   $\hspace{10pt}\tilde{\mathcal U}_v = \tilde{\mathcal U}_v\ominus \varepsilon_1 $\ \label{step: Ua reduction}
%   $I = \tilde I$;
   }
}
 \lIf {$\tilde{\mathcal U}_v = \emptyset$}{
 
 $\hspace{10pt}c = c+\varepsilon_2$;\\
 $\hspace{10pt}\tilde{\mathcal U}_v = \mathcal U_v$}
 }
 $N_p = 2\; N_p$; \\
 $c = 0$;
 }
 \textbf{Return: }$\bar {\mathcal U}_v, c$\;
 \caption{Iterative method for computing $\tilde{\mathcal U}_v, c$}
\end{algorithm}

\begin{Remark}\label{rem: attack model 1}
If it is unknown which components of the input are vulnerable, then all possible combinations of $u_v$ and $u_s$ can be considered, and Algorithm \ref{algo: iter U c comb} can be used to compute $c$ for each such combination. Then, the maximum of all such values can be used to define the set $S_c$, guaranteeing the system's security against attack on any control inputs. 
% For example, in the case of four inputs (i.e., $m = 4$), there are a total of fourteen such possible combinations.  
\end{Remark}

\begin{Remark}
The computational complexity of Algorithm \ref{algo: iter U c comb} is only a function of the number of sampling points $N_p$ (which, in principle, is a user-defined parameter) and is independent of the non-linearity of the function $F$, and linear in the dimension $n$. Thus, unlike reachability based tools in \cite{bansal2017hamilton,choi2021robust} where the computational complexity grows exponentially with the system dimension $n$, or SOS based tools \cite{wang2018permissive} that are only applicable to a specific class of systems with linear or polynomial dynamics, Algorithm \ref{algo: iter U c comb} can be used for general nonlinear system with high dimension.  
\end{Remark}
So far, we presented sufficient conditions to establish the safety of the system \eqref{eq: actual system} under attacks (Proposition \ref{lemma: sufficient conditions non-attack}), a sampling-based method to verify these conditions using a finite number of sampling points (Theorem \ref{thm: sufficient sup H cond}), and iterative methods to compute the set of initial conditions and the input constraint set to satisfy these conditions (Algorithm \ref{algo: iter U c comb}). Thus, in brief, using the results in this section, we can compute the viability domain $S_c$ and control input constraint set $\tilde{\mathcal U}_v\subset\mathcal U_v$, such that for all $x\in S_c$ and $u_v\in \tilde{\mathcal U}_v$, there exists a control input $u_s\in \mathcal U_s$ that can keep the system trajectories in the set $S_c$ at all times. In the next section, we present a method of computing such a control input using a QP formulation.
% , thus, providing the final piece of the solution to Problem \ref{Problem 1}.}

\section{QP based feedback design}\label{sec: QP control}
In this section, we use the sufficient conditions from the previous section to design a feedback law for the system \eqref{eq: actual system} that guarantees security with respect to the safety property under Assumption \ref{assum: d bound}. 
% \subsection{QP based approach}
% For the results in this section, 
We assume that the control input constraint set is given as $\tilde{\mathcal U}\coloneqq\tilde{\mathcal U}_v\times\mathcal U_s =\{v\in \mathbb R^{m}\; |\; u_{j,min}\leq v_j\leq u_{j,max}\}$, i.e., as a box-constraint set where $u_{j,min}< u_{j,max}$ are the lower and upper bounds on the individual control inputs $v_j$ for $j = 1, 2, \ldots, m$, respectively. We can write $\mathcal U$ in a compact form as $\tilde{\mathcal U} = \{v\; |\; A_{u}v\leq b_{u}\}$ where $A_u\in \mathbb R^{2m\times m}, b_u\in \mathbb R^{2m}$. Furthermore, we assume that the system model \eqref{eq: actual system} is control affine, and is of the form:
\begin{align}\label{eq: cont affine}
    \dot x & = f(x) + g_v(s)u_v+g_s(x)u_s + d(t,x),
\end{align}
where $f:\mathbb R^n\rightarrow\mathbb R^n$, $g_v: \mathbb R^n\rightarrow\mathbb R^{n\times m_v}$ and $g_s: \mathbb R^n\rightarrow\mathbb R^{n\times (m-m_v)}$ are continuous functions. In this case, the function $H:\mathbb R^n\times\mathbb R^{m_v}\to\mathbb R$ reads{\small
\begin{align}
    \hspace{-5pt}H(x,u_v) = \hspace{-3pt}\inf_{u_s\in \mathcal U_s}\hspace{-5pt}L_fB(x) + L_{g_s}B(x)u_s +L_{g_v}B(x)u_v.
\end{align}}\normalsize
% Similarly, under the second class of attack on actuators, as described through \eqref{eq: attack model additive attack}, the special case of the control affine model takes the following form:
% \begin{align}\label{eq: NL pert cont affine  additive attack}
%     \dot x(t) & = f(x(t)) + \begin{bmatrix}g(x(t)) &  g(x(t))\end{bmatrix}\begin{bmatrix}u(t) \\ \tilde u_v(t)\end{bmatrix},
% \end{align}
% and the function $H$ reads
% \begin{align}
%     H(x,u_v) = \inf_{u\in \mathcal U\ominus \tilde{\mathcal U}_v}L_fB(x) + L_{g}B(x)u +L_gh_{S}(x)\tilde u_v.
% \end{align}
% \normalsize 
In addition to the safety requirement in Problem \ref{Problem 1}, we impose the requirement of convergence of the system trajectories of \eqref{eq: cont affine} to the origin. To this end, given a twice continuously differentiable, positive definite function $V:\mathbb R^n\rightarrow \mathbb R_+$ as a candidate Lyapunov function, we use the condition 
\begin{align}
   \hspace{-10pt}L_fV(x) + L_{g_s}V(x)u_s+L_{g_v}V(x)u_v \leq -\zeta V(x) - l_V \delta,
\end{align}
where $\zeta>0$, to guarantee convergence of the system trajectories to the origin under $d$ satisfying Assumption \ref{assum: d bound}. We assume that the set $S$ is an $(n-1)$-unit sphere, so that we can use the results from the previous section to compute a viability domain for it, and that $0\in \textnormal{int}(S)$, so that the convergence requirement is feasible. The linear constraints on the control input, and the system model being control affine, helps us formulate a convex optimization problem that can be efficiently solved for real-time control synthesis \cite{ames2017control}. We propose the following Quadratic Program (QP) to solve Problem \ref{Problem 1}. 
% where the objective is to minimize the norm of the control input \cite{garg2019prescribedTAC}. 
Define $z =(v_s, v_v, \eta, \zeta)\in \mathbb R^{m+2}$ and for a given $x\in \mathbb R^n$, consider the following {QP}:\vspace{-10pt}

\small{
\begin{subequations}\label{QP gen}
\begin{align}
\hspace{10pt}\min_{z} \quad \frac{1}{2}|z|^2 + q\zeta\\
    \textrm{s.t.} \;\quad \; \; A_uv_{na}  \leq &  \; b_u, \label{C1 cont const}\\
    % L_fV(x) + L_gV(x)v  \leq & \; \delta_1V(x)-\alpha V(x)^{\gamma_1}\nonumber\\
    % & -\alpha V(x)^{\gamma_2} -l_G\delta,\label{C2 stab const}\\
    L_fB(x) + L_{g_s}B(x)v_s \leq & -\eta ~(B(x)+c)\nonumber \\
    & -\sup_{u_v\in \tilde{\mathcal U}_v}L_{g_v}B(x)u_v-l_B\delta,\label{C3 safe const}\\
     L_fV(x) + L_{g_s}V(x)v_s  + & L_{g_v}V(x)v_v\leq  -\zeta ~ V(x) -l_V\delta,\label{C4 conv const}
    % \\
    % \frac{2}{T_{ud}}\leq \alpha_1(\gamma_1-1) \leq &  \alpha_2(1-\gamma_2)\label{C4 T const}, 
    % \frac{\mu\pi}{2T_{ud}}\leq & \alpha_1, \label{C4 T const}\\  \frac{\mu\pi}{2T_{ud}}\leq & \alpha_2, \label{C4 T const2}
\end{align}
\end{subequations}}\normalsize
where $q>0$ is a constant, $l_B, l_v$ are the Lipschitz constants of the functions $B$ and $V$, respectively, and $c$ and $\tilde{\mathcal U}_v$ are the output of Algorithm \ref{algo: iter U c comb}. Here, $\eta$ and $\zeta$ are slack variables used for guaranteeing feasibility of the QP (see \cite[Lemma 6]{garg2019prescribedTAC}). The first constraint \eqref{C1 cont const} is the input constraints, the second constraint is the CBF condition from Lemma \ref{lemma: nec suff safety} for forward invariance of the set $S_c$ and the third constraint \eqref{C4 conv const} is CLF constraint for convergence of the system trajectories to the origin. Note that the secure input $v_s$ is used in both \eqref{C3 safe const} and \eqref{C4 conv const}, while the vulnerable input $v_v$ is only used in \eqref{C4 conv const}. 
% The parameters $\alpha, \alpha, \gamma_1, \gamma_2$ are fixed, and are chosen as $\alpha = \frac{\mu\pi}{2T_{ud}}$, $\gamma_1 = 1+\frac{1}{\mu}$ and $\gamma_2 = 1-\frac{1}{\mu}$ with $\mu>1$. 

Let the optimal solution of \eqref{QP gen} at a given point $x\in \mathbb R^n$ be denoted as $z^*(x) = (v_s^*(x), v_v^*(x), \eta^*(x),\zeta^*(x))$. In order to guarantee continuity of the solution $z^*$ with respect to $x$, we need to impose the strict complementary slackness condition on \eqref{QP gen} (see \cite{garg2019prescribedTAC}). In brief, if the $i-$the constraint of \eqref{QP gen}, with $i \in \{1, 2, 3\}$, is written as $G_i(x,z)\leq 0$, and the corresponding Lagrange multiplier is $\lambda_i\in \mathbb R_+$, then strict complementary slackness requires that $\lambda_i^*G(x,z^*)<0$, where $z^*, \lambda_i^*$ denote the optimal solution and the corresponding optimal Lagrange multiplier, respectively. We are now ready to state the following result.  

% \begin{Theorem}
% Let $c_1$ and $c_2$ be the output of the Algorithms \ref{algo iter c} and Algorithm \ref{algo iter c conv} corresponding to the functions $H_s$ and $H_G$ for the system \eqref{eq: NL pert cont affine} (respectively, \eqref{eq: NL pert cont affine  additive attack}), and assume that there exists a control input $u_s$ such that the sets $S_c = \{x\; |\; B(x)\leq -c_1\}$ and $D_c =  \{x\; |\; B(x)\leq \eta-c_2\}$ forward-invariant for the model \eqref{eq: NL pert cont affine} (respectively, \eqref{eq: NL pert cont affine  additive attack}). Then, the QP \eqref{QP gen} is feasible for all $x\in S_c\cap D_c$ and the control input $u = v^*$, where $v^*$ is the solution of the QP \eqref{QP gen} leads to security of \eqref{eq: actual system} for safety and convergence properties under an attack of the form \eqref{eq: model} (respectively, \eqref{eq: attack model additive attack}).
% \end{Theorem}

\begin{Theorem}
Given the functions $F, d, B, V$ and the attack model \eqref{eq: attack model}, suppose Assumptions \ref{assum: d bound}-\ref{assum: max dist tring nD} hold. Let $c$ and $\tilde{\mathcal U}_v$ be the output of the Algorithm \ref{algo: iter U c comb}. Assume that the strict complementary slackness holds for the QP \eqref{QP gen} for all $x\in S_c$. Then, the QP \eqref{QP gen} is feasible for all $x\in S_c$, and the control law defined as $k_s(x) = v_s^*(x)$ is continuous on $\textnormal{int}(S_c)$, and solves Problem \ref{Problem 1} for all $x(0)\in X_0\coloneqq  \textnormal{int}(S_c)$.
% i.e., the resulting closed-loop system $\mathcal S$ satisfies $[\mathcal A_\tau]\mathcal S\models P_{safe, \textnormal{int}(S_c)}$ for all $\tau\geq 0$. 
\end{Theorem}

\begin{proof}
% The proof is based on the results in \cite{garg2019prescribedTAC}.
Per Theorem \ref{thm: sufficient sup H cond}, the set $S_c$ is a viability domain for the system \eqref{eq: cont affine} under Assumption \ref{assum: LC F del B}. Thus, feasibility of the QP \eqref{QP gen} follows from \cite[Lemma 6]{garg2019prescribedTAC}. Note also that with $V$ being twice continuously differentiable and under Assumption \ref{assum: LC F del B}, the Lie derivatives of the functions $V$ and $B$ along $f, g_s$, and $g_v$ are continuous. Thus, per \cite[Theorem 1]{garg2019prescribedTAC}, the solution $z^*$ of the QP \eqref{QP gen} is continuous on $\textnormal{int}(S_c)$. Finally, since the set $S_c$ is compact, it follows from \cite[Lemma 7]{garg2019prescribedTAC} that the closed-loop trajectories are uniquely defined for all $t\geq 0$. Uniqueness of the closed-loop trajectories, Assumption \ref{assum: d bound} and feasibility of the QP \eqref{QP gen} for all $x\in S_c$ implies that all the conditions of Lemma \ref{lemma: nec suff safety} are satisfied and it follows that the set $S_c$ is forward invariant for the system \eqref{eq: cont affine}. 
% , i.e., the closed-loop system $\mathcal S$ satisfies $[\mathcal A_\tau]S\models P_{safe, \textnormal{int}(S_c)}$ for all $\tau\geq 0$.
\end{proof}

{
\begin{Remark}
In this work, only the control input $u_s$ is used to achieve safety since it is unknown when the vulnerable input $u_v$ comes under an attack. This conservative assumption can be relaxed by utilizing an attack-detection mechanism, which can trigger a switching mechanism from a \textnormal{nominal} control design, assuming no attacks, to the proposed method under an attack. We leave this detection-based switching mechanism as part of our future work.
\end{Remark}
}

\section{Numerical Experiments}
% \subsection{3D Linear-Time Invariant system example}
To showcase the effectiveness of the proposed method, we present an academic example with the system given as 
\begin{align}
    \dot x = f(x) + Ax + Bu + d(t,x),
\end{align}
where $A\in \mathbb R^{3\times 3}$ and $B\in \mathbb R^{3\times 2}$. The input constraint sets are $\mathcal U_1 = \{u_1\in \mathbb R\; |\; |u_1|\leq u_{M1}\}$ and $\mathcal U_2 = \{u_2\in \mathbb R\; |\; |u_2|\leq u_{M2}\}$ for some $u_{M1}, u_{M2}>0$. The safe set is $S = \{x\in \mathbb R^3\; |\; |x|^2-1\leq 0\}$ corresponding to the function $h(x) = |x|^2-1$, i.e., the safe set is the unit sphere. We use randomly generated matrices $A$ and $B$ such that the pairs $(A,B_1)$ and $(A, B_2)$ are controllable, where $B_1$ and $B_2$ are the first and the second columns of the matrix $B$, respectively. The matrices $(A,B)$ and the function $f$ are

{\footnotesize
\begin{align*}
    A \!=\!\begin{bmatrix}0.61  & \hspace{-12pt} 0.37 & \hspace{-12pt} 2.69\\
    -0.06& \hspace{-12pt} -1.02 & \hspace{-12pt} -0.88\\
    1.33& \hspace{-12pt} -2.71& \hspace{-12pt} 0.91\end{bmatrix}\!, B\!=\! \begin{bmatrix}-0.24 & \hspace{-10pt} 0.04\\
    0.32 & \hspace{-10pt} -0.01\\
    -1.12& \hspace{-10pt}-0.07\end{bmatrix}\!, f(x) \!=\! 0.01\!\begin{bmatrix}x_1^3 + x_2^2x_3\\ x_2^3 + x_3^2x_1\\x_3^3 + x_1^2x_2 \end{bmatrix}.
\end{align*}}\normalsize
% \begin{align*}
% \end{align*}

\begin{figure}[b]
\vspace{-10pt}
	\centering
	\includegraphics[width=\columnwidth,clip]{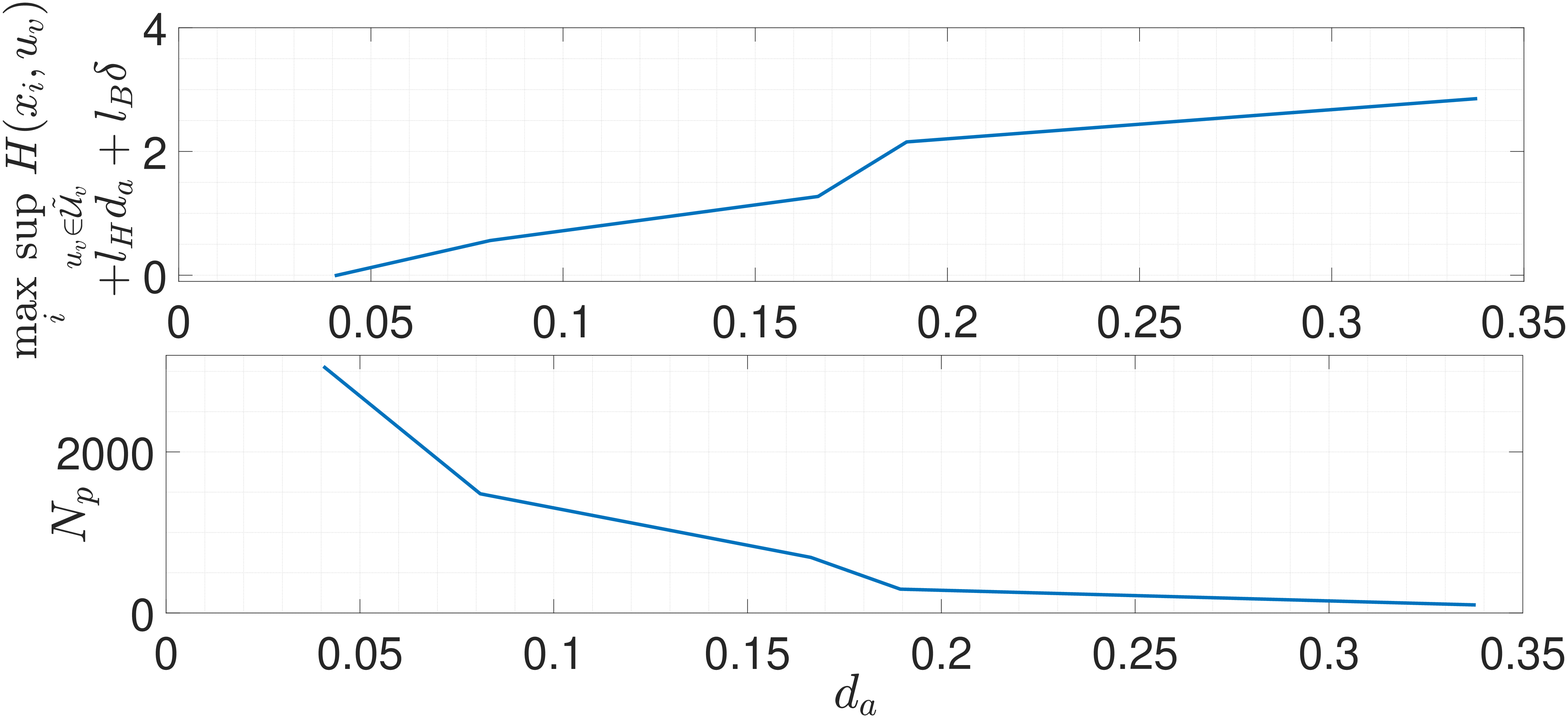}
	\caption{The value of $\max_i\sup_{u_v\in \tilde{\mathcal U}_v}H(\bar x,u_v) +l_Hd_a +l_B\delta+l_Hd_a$ and the number of sampling points for different values of $d_a$.}
	\label{fig:sup H}
\end{figure}
We use MATLAB code from \cite{anton2021SamplingCode} to generate a uniform sampling on the boundary of the unit sphere.
% , as shown in Figure \ref{fig: sphere mesh}. 
Figure \ref{fig:sup H} shows the maximum value of $\sup_{u_v\in \tilde{\mathcal U}_v}H(\bar x,u_v)+l_Hd_a +l_B\delta$ over the sampling points for different values of $d_a$. It is observed that condition \eqref{eq: sup H cond 3D} is satisfied when $d_a =0.0406$, and the corresponding number of sampling points is $N_p = 3062$. 
% \begin{figure}[b]
% 	\centering
% 	\includegraphics[width=0.9\columnwidth,clip]{h_plot_ex1_case0.eps}
% 	\caption{Evolution of the function $h$ with time.}
% 	\label{fig:ex1case1 h func}
% \end{figure}
% \begin{figure}[t]
% 	\centering
% 	\includegraphics[width=0.9\columnwidth,clip]{u_plot_ex1_case0.eps}
% 	\caption{vulnerable input $u_1$ and the control input $u_2$ with time.}
% 	\label{fig:ex1case1 input}
% \end{figure}
% \begin{figure}[t]
% 	\centering
% 	\includegraphics[width=1\columnwidth,clip]{h_plot_ex1_case0_1.eps}
% 	\caption{Evolution of the function $h$ with time.}
% 	\label{fig:ex1 h func}
% \end{figure}
Without loss of generality, we assume that $u_2$ is vulnerable.

\begin{figure}[t]
	\centering
	\includegraphics[width=0.95\columnwidth,clip]{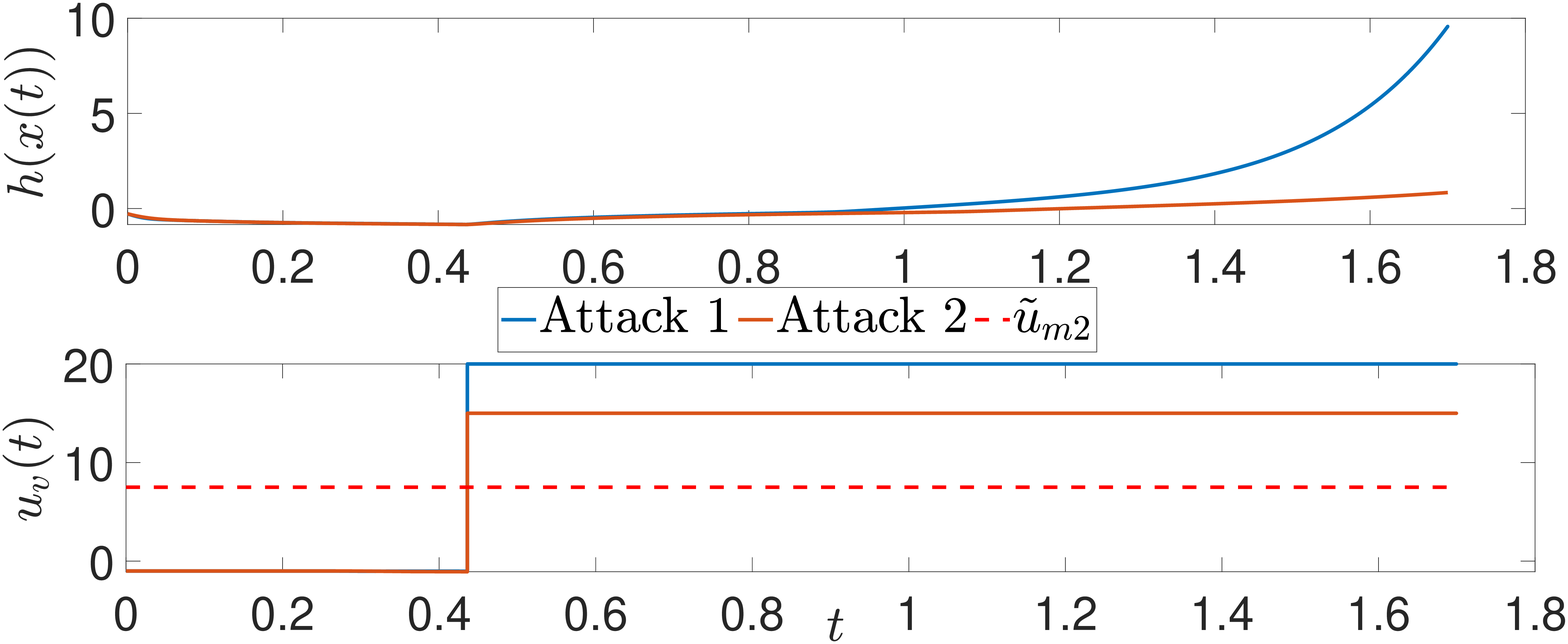}
	\caption{The vulnerable input $u_v$ and the function $h$ under attacks 1 and 2.}
\vspace{-10pt}
	\label{fig:ex1 uh func 10 5}
\end{figure}
\begin{figure}[t]
	\centering
	\includegraphics[width=0.95\columnwidth,clip]{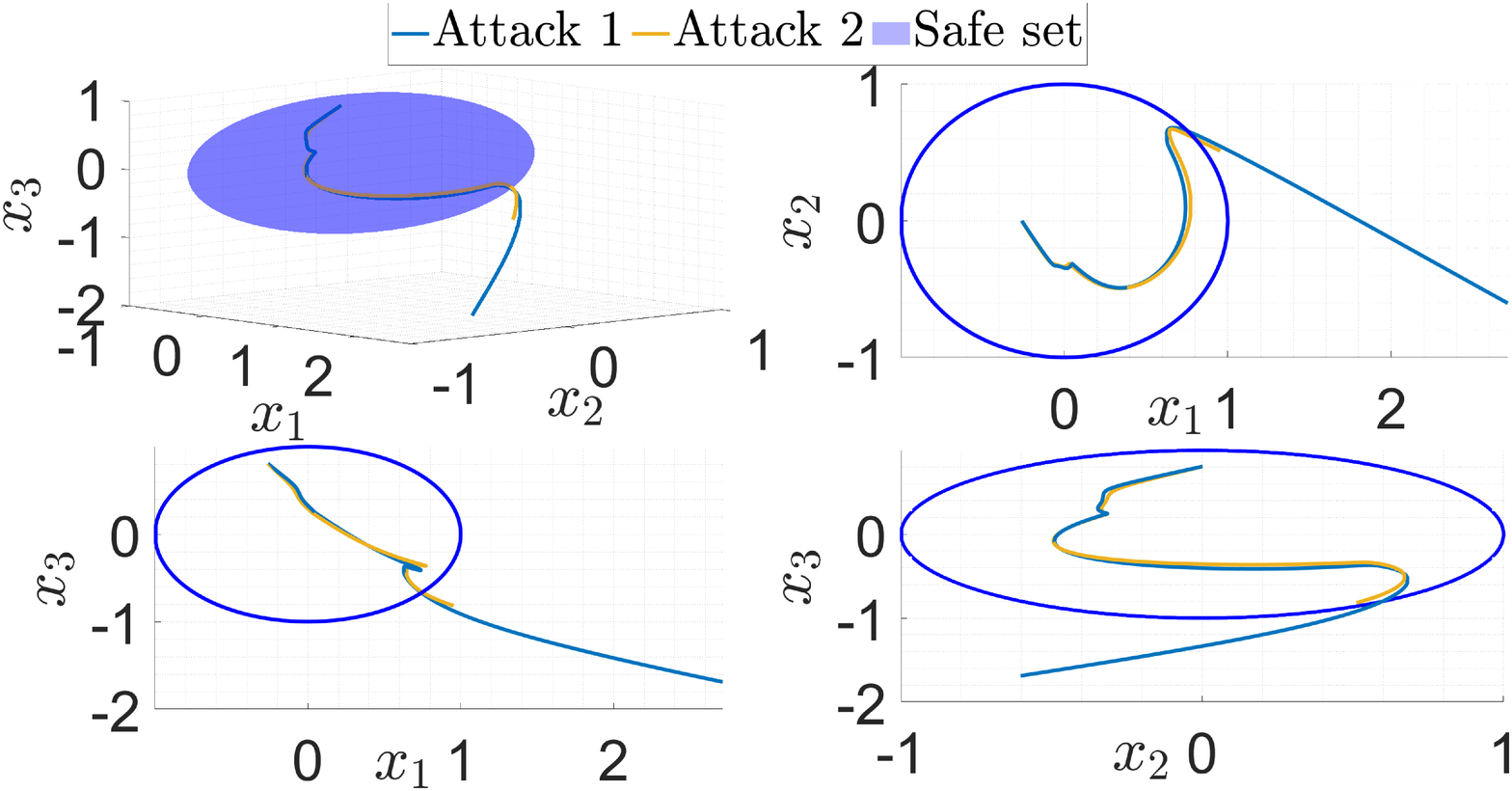}
	\caption{The closed-loop paths traced by the system under attacks 1 and 2. 
% 	The top figure of is a 3D plot, while the rest of the figures are projections on $x_1-x_2, x_1-x_3$ and $x_2-x_3$ planes.
	}
\vspace{-20pt}
	\label{fig:ex1 traj 10 5}
\end{figure}

\noindent We use Algorithm \ref{algo: iter U c comb} to compute the set $\tilde{\mathcal U}_i$ and a value of $c$ such that \eqref{eq: sup H cond 3D} holds for all the sampling points. With $u_{M1} = 20$ and $u_{M2} = 20$ (defining the sets $\mathcal U_s, \mathcal U_v$), Algorithm \ref{algo: iter U c comb} gives $c = 0$ for the viability domain $\{x\; |\; h(x)\leq c\}$ and $\tilde u_{M2} = 7.5$ (defining the set $\tilde{\mathcal U}_v$) as the feasible bound on the attack signal $u_2$. The attack happens at a randomly chosen $\tau = 0.436$ and $\delta = 0.1$ in Assumption \ref{assum: d bound}. 
% {\footnotesize
% \begin{align*}
%     A = \begin{bmatrix}0.4154 & 0.0683 & -0.4019\\
%     0.5773  &  -0.4196  &  0.8921\\
%     0.5684  & -0.1941  &  1.2891\end{bmatrix}, B = \begin{bmatrix}4.5121  & -8.5832\\
%     2.5132 & -10.0804\\
%   -12.1319  &  7.1836\end{bmatrix}
% \end{align*}}\normalsize
% with 

First, we illustrate that the system violates safety when the attack signal $u_2$ does not satisfy the bounds computed by Algorithm \ref{algo: iter U c comb}. Figure \ref{fig:ex1 uh func 10 5} shows the vulnerable input $u_v$ for the initial two attack scenarios (Attack 1 and 2) where $\bar u_{M2} = 20$ and $\bar u_{M2} = 15$, i.e., the set $\tilde{\mathcal U}_v$ is larger than the one computed using the proposed algorithm. Figure \ref{fig:ex1 uh func 10 5} also plots the evolution of the barrier function $h$ with time for the two cases. It can be observed that the function $h$ corresponding to this attack takes positive values, and thus, the safety property for the system is violated. Figure \ref{fig:ex1 traj 10 5} plots the corresponding closed-loops paths for the two scenarios, and it can be seen that the system leaves the safe set, thus violating safety. 

In the rest of the attack scenarios (Attack 3-6), the bound $|u_v|\leq 7.5$ is imposed as computed by the proposed algorithm.  Figure \ref{fig:ex1 uh func rest} plots the different types of attack signals used in these scenarios, namely, saturated signals with $u_v = 7.5$ and $u_v = -7.5$, square wave and sinusoidal signal, both with amplitude $7.5$. The corresponding evolution of the barrier function $h$ illustrates that the system maintains safety in all four scenarios. Figure \ref{fig:ex1 traj rest} plots the closed-loops paths for these attack scenarios, and it can be seen that the system trajectories evolve in the safe set at all times, thus maintaining safety. 
%Finally, to illustrate the effect of an attack on systems' performance, we simulated without any attacks. Figure \ref{fig: V plot} plots the function $V(x) = \frac{1}{2}|x|^2$ for all the attack scenarios as well as when there is no attack on $u_v$. It can be seen that the value of the function $V$ under an attack is higher than the case when there is no attack, illustrating that an attack on the system input $u_v$ can deteriorate its performance. 
% Another line of future work is an initial characterization of a class of systems where our tools are applicable and others where it fails (where there are no bounds that can prevent successful attacks). 
%In this work, we assumed that the time when the vulnerable input can be attacked is unknown. Future work involves devising an attack-detection mechanism and switching controller to propose more efficient controllers.
Through this case study, we illustrate that if the system parameters are not chosen according to our proposed method, then there might exist attacks that can lead to violation of safety. On the other hand, when the system parameters are designed according to the proposed algorithm, no attack can violate safety, confirming that the system is secure by design. 
%, and it can be observed that $h(x(t))\leq 0$ for all $t\geq 0$ in both cases, implying that the system trajectories evolve in the safe set at all times. Figure \ref{fig:ex1 input} plots the input and vulnerable input for the two cases, and it can be seen that the inputs satisfy the imposed bounds in both cases. 

\begin{figure}[t]
	\centering
	\includegraphics[width=0.95\columnwidth,clip]{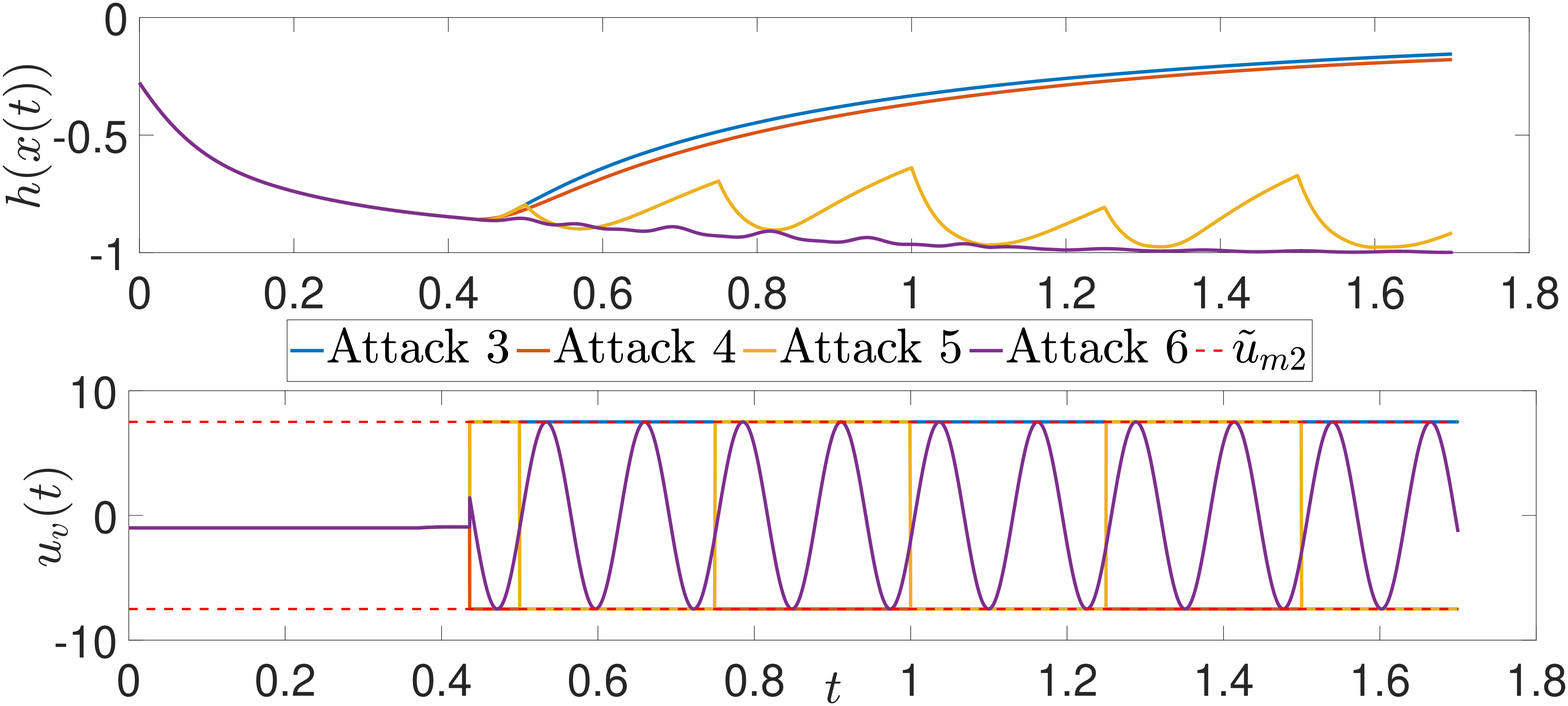}
	\caption{The vulnerable input $u_v$ and the function $h$ under attacks 3-6.}
\vspace{-10pt}
	\label{fig:ex1 uh func rest}
\end{figure}
\begin{figure}[t]
	\centering
	\includegraphics[width=0.95\columnwidth,clip]{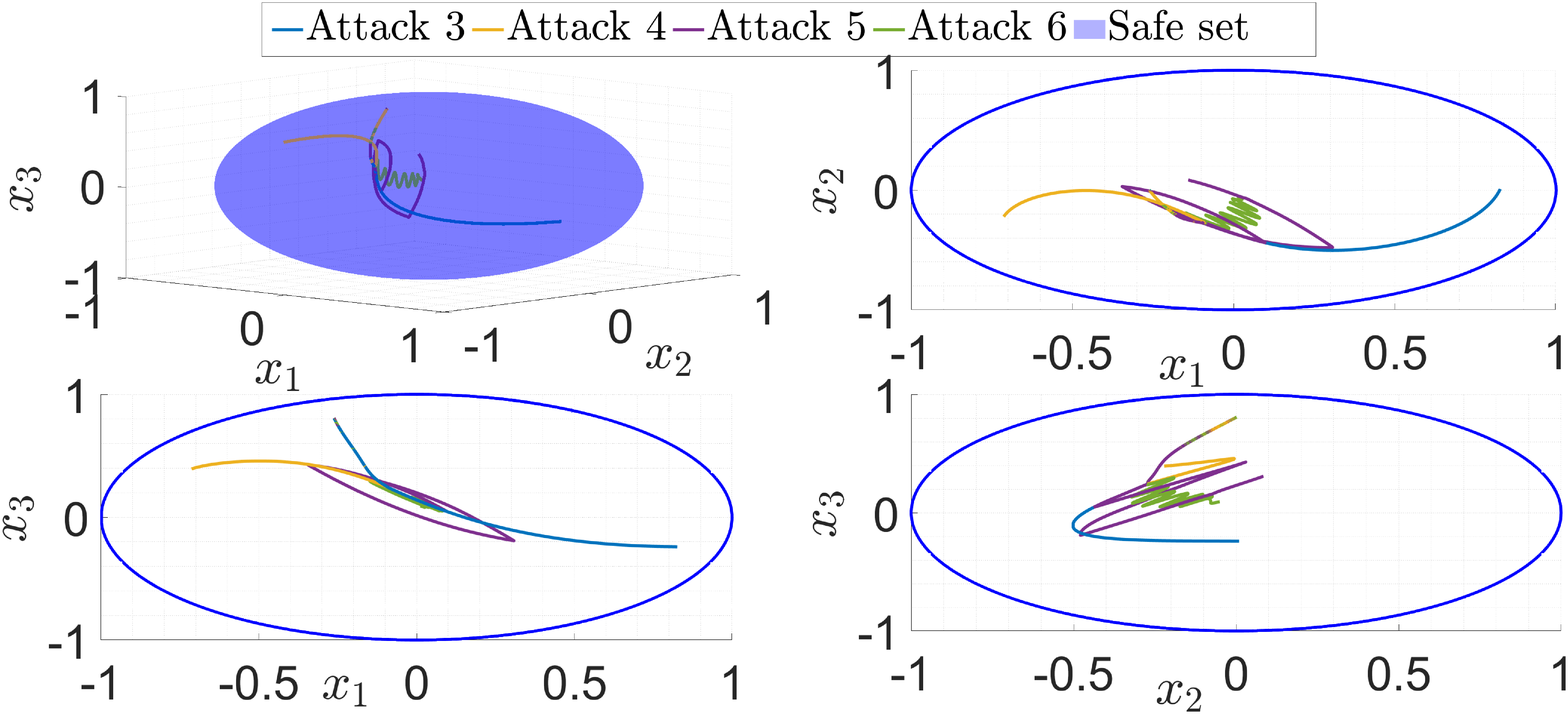}
	\caption{The closed-loop paths traced by the system under attacks 3-6. 
% 	The top figure of is a 3D plot, while the rest of the figures are projections on $x_1-x_2, x_1-x_3$ and $x_2-x_3$ planes. 
	}
\vspace{-20pt}
	\label{fig:ex1 traj rest}
\end{figure}

\section{Conclusion and Future Work}
{In this paper, we study the problem of computing a viability domain and input constraint set so that the safety of a system can be guaranteed under attacks on the system inputs. In contrast to prior work on the computation of viability domain whose applicability is limited to linear or polynomial dynamics or whose computational complexity grows exponentially with system dimension, our method is computationally efficient and applies to a general class of nonlinear systems. We showed that when the system parameters are chosen using our sampling-based iterative algorithm, the resulting system is resilient to arbitrary attacks, and thus, is secure by design.}
% \textbf{Discussion and Future Work:} 

Our approach can be used to design bounds (that can either be implemented physically or in a tamper-proof reference monitor) that will prevent attackers from driving control systems to unsafe states. It is efficient (sampling-based viability computation) and general (applicable to non-linear systems).
% This line of work opens up several new problems that we will consider in future work, such as the operational impact of the proposed strategy. 
By limiting the range of actuation and the initial set, we are limiting the responsiveness of control action, and in general, systems with our defense might converge slower to the desired set point or trajectory. One way to mitigate this is to use attack-detection mechanisms and switching strategy so that more efficient controllers can be used when the system is not under an attack. 
% We considered systems modeled under a general class of nonlinear dynamics; thus, the proposed results can be used for a large class of real-world systems. Then, we showed that it is possible to design computationally tractable iterative algorithms to compute viability domain for 2D and 3D systems. Finally, we showed that the computed viability domain could be used in a provably feasible QP formulation for the real-time synthesis of feedback law.

\bibliographystyle{IEEEtran}
\bibliography{myreferences}

\end{document}